\documentclass[12pt,onecolumn,draftclsnofoot]{IEEEtran}

\usepackage{amsmath}
\usepackage{amssymb}
\usepackage{amsthm}
\usepackage{graphicx}
\usepackage{epstopdf}
\usepackage{bm}
\usepackage{color}
\usepackage{cite}
\usepackage{subcaption}
\usepackage{dsfont}
\usepackage{algorithm}
\usepackage{algpseudocode}
\usepackage{enumitem}

\newtheorem{theorem}{Theorem}
\newtheorem{lemma}[theorem]{Lemma}

\newtheorem{remark}{Remark}

\allowdisplaybreaks

\begin{document}

\title{Estimation of Scalar Field Distribution in the Fourier Domain}
\author{Alex S.\ Leong and Alexei T. Skvortsov \\  
 \thanks{The authors are with Defence Science and Technology Group, Melbourne, Australia. E-mail: {\tt alex.leong@defence.gov.au, alexei.skvortsov@defence.gov.au} }  
  }

\maketitle

\begin{abstract}
%In this paper we consider the problem of estimation of scalar field distribution collected from noisy measurements. The field is modelled as a sum of Fourier components/modes, where the number of modes retained and estimated determines in a natural way the approximation quality. An algorithm for estimating the modes using an online optimization approach is presented, under the assumption that the noisy measurements are quantized. The algorithm can estimate time-varying fields through the introduction of a forgetting factor. Simulation studies demonstrate the effectiveness of the proposed approach. 

In this paper we consider the problem of estimation of scalar field distribution (e.g. pollutant, moisture, temperature) from noisy measurements collected by unmanned autonomous vehicles such as UAVs. The field is modelled as a sum of Fourier components/modes, where the number of modes retained and estimated determines in a natural way the approximation  quality. An algorithm for estimating the modes using an online optimization approach is presented, under the assumption that the noisy measurements are quantized. The algorithm can also estimate time-varying fields through the introduction of a forgetting factor. Simulation studies
demonstrate the effectiveness of the proposed approach.
\end{abstract}

\section{Introduction}

%Estimation of scalar field distribution from a set of point measurements is an important problem often emerging in domains such as atmospheric monitoring, ecology, and geophysics. 
Estimation of scalar field distribution from a set of point measurements  is an important problem often emerging in domains such as atmospheric monitoring, risk assessment, and  hazard mitigation.
Examples include concentration of pollutant, carbon dioxide emission, methane sources, radiation, temperature in urban areas, and many others, see 
\cite{HutchinsonOh, RisticMorelandeGunatilaka, Yardibi, AnnunzioYoungHaupt, NeumannBennetts_advanced_robotics, WadeSenocak, NewazJeong, RisticGunatilakaGailis, Selvaratnam_CDC, HutchinsonLiu, EslingerMendez, LiChen, ParkAnSeoOh, WeidmannHirst,  MartinPayton, WangYangWu, MorelandeSkvortsov,  LaSheng, LaShengChen, RazakSukumarChung_journal, LeongZamani_SP, LeongZamaniShames, TranGarratt} and the references therein. This approach is often used for indirect inference of scalar fields (pollutant concentration, pressure, temperature, radiation)  in inaccessible locations where the direct measurements are prohibited due to some geometrical or physical constraints (blocking obstacles, high temperature, or exposure to hazards). The methods of source localisation \cite{HutchinsonOh, RisticMorelandeGunatilaka, Yardibi, AnnunzioYoungHaupt, NeumannBennetts_advanced_robotics, WadeSenocak, NewazJeong, RisticGunatilakaGailis, Selvaratnam_CDC, HutchinsonLiu, EslingerMendez, LiChen, ParkAnSeoOh, WeidmannHirst} and mapping \cite{MartinPayton, WangYangWu, MorelandeSkvortsov,  LaSheng, LaShengChen, RazakSukumarChung_journal, LeongZamani_SP, LeongZamaniShames, TranGarratt} employing remote (and noisy) measurements have attracted increasing attention in recent years due to tremendous progress in instrumentation for aerial and remote sensing using unmanned autonomous vehicles such as UAVs. This technological advancement necessitates the development and evaluation of some statistical methods and algorithms that can be applied for the timely estimation of the structure (map) of the scalar field in the environment from an ever-increasing set of noisy measurements acquired in a  sequential or concurrent manner (e.g.,  sensing signals from unmanned vehicles  operating over the hazardous area, trigger signals from meteorological stations, the intermittent concentration of methane emissions in the atmosphere or ocean floor, oil surface concentration due to androgenic spill,  etc). These algorithms may become critical for backtracking and characterization of the main sources of the scalar field in the environment which is important for the  remediation effectiveness and retrospective forensic analysis.   This was the main motivation for the present study.

In work on estimation of scalar fields, the field is often modelled as a sum of radial basis functions (RBFs) or Gaussian mixture models, see, e.g., \cite{MorelandeSkvortsov,  LaSheng, LaShengChen, RazakSukumarChung_journal, LeongZamani_SP, LeongZamaniShames, TranGarratt}. Field estimation then reduces to a problem of estimating the parameters of these models. By contrast, for the current work, we assume the field to be an arbitrary 2D function which can be viewed in the Fourier domain using, e.g., the discrete Fourier transform (DFT) or the discrete cosine transform (DCT) \cite{BritanikYipRao}. For some intuition behind this approach, suppose we  regard the plot of the field as an image. From image processing, it is well-known that the most important parts of an image are concentrated in the lowest (spatial) frequency components/modes. Our approach to field estimation is then to estimate the low frequency Fourier components.\footnote{We will use the terms Fourier component and DCT component interchangeably in this paper.} One of the advantages for using this Fourier component approach compared to the RBF approach is that it offers a perhaps more natural way to control the accuracy of the approximation, e.g., by controlling the number of Fourier modes used/retained. Furthermore, if one wants to refine the field estimate by estimating more modes, existing estimates of the lower order modes can be reused. 

The main contributions of this paper are:
\begin{itemize}
    \item Rather than the use of radial basis function field models, we model the 2D scalar field in the Fourier domain as a sum of Fourier components. 
    \item A numerical comparison of the approximation capabilities of the Fourier components and RBF field models is carried out.    
    \item For the quantized measurements model, we present in detail how Fourier component estimation can be carried out using an online optimization approach similar to \cite{LeongZamaniShames}. We further extend the approach of \cite{LeongZamaniShames} from binary measurements to multi-level quantized measurements, and from static to time-varying fields. 
\end{itemize}

The organization of the paper is as follows: Section \ref{sec:preliminaries} gives preliminaries on the DCT and motivation for its use in field modelling. Section \ref{sec:system_model} presents the system model. Section~\ref{sec:DCT_RBF_comparison} compares our Fourier component field model with the RBF field model in terms of approximation performance. Section \ref{sec:DCT_estimation} %first relates our field model to the RBF field models considered in previous works, and then 
considers in detail the estimation of Fourier components using  quantized measurements. Numerical studies  are presented in Section \ref{sec:numerical}. 

\section{Preliminaries}
\label{sec:preliminaries}

Consider a region of interest $\mathcal{S} = [X_{\textnormal{min}}, X_{\textnormal{max}}] \times [Y_{\textnormal{min}}, Y_{\textnormal{max}}]$. Discretize $[X_{\textnormal{min}}, X_{\textnormal{max}}]$ into $N_x$ points and $ [Y_{\textnormal{min}}, Y_{\textnormal{max}}]$ into $N_y$ points as
\begin{align*}
\mathcal{X}_d \triangleq \left\{X_{\textnormal{min}} + \Big(\frac{1}{2} + I_x \Big) \Delta_x:   I_x \in \{0, \dots, N_x - 1\} \right\} \\
\mathcal{Y}_d \triangleq \left\{Y_{\textnormal{min}} + \Big(\frac{1}{2} + I_y \Big) \Delta_y:  I_y \in \{0, \dots, N_y - 1\} \right\},
\end{align*}
where 
$$\Delta_x \triangleq \frac{X_{\textnormal{max}} - X_{\textnormal{min}}}{N_x}, \quad \Delta_y \triangleq \frac{Y_{\textnormal{max}} - Y_{\textnormal{min}}}{N_y}. $$

Our aim is the estimation of 2D distribution of a scalar field $\phi(x,y)$, $(x,y) \in \mathcal{S}$,  which is assumed either static or slowly varying. We define
$$\phi_d(I_x, I_y) \triangleq \phi\Big(X_{\textnormal{min}} + \left(1/2 + I_x \right) \Delta_x, Y_{\textnormal{min}} + \left(1/2 + I_y \right) \Delta_y \Big)$$
as the field value at the discretized position $\big(X_{\textnormal{min}} + \left(1/2 + I_x \right) \Delta_x, Y_{\textnormal{min}} + \left(1/2 + I_y \right) \Delta_y \big) \in \mathcal{X}_d \times \mathcal{Y}_d$, where $(I_x, I_y)$ can be regarded as a position index. 
Recall the (Type-II) discrete cosine transform (DCT), see, e.g., \cite{BritanikYipRao,Strang_DCT}: 
\begin{align*}
C(u,v) &= \sum_{I_x=0}^{N_x-1} \sum_{I_y=0}^{N_y-1} \alpha_x(u) \alpha_y(v) \phi_d(I_x,I_y) \cos \left(\frac{(2I_x+1)\pi u}{2 N_x} \right) \cos \left(\frac{(2I_y+1)\pi v}{2 N_y} \right), \\ &\quad\quad u=0,\dots,N_x-1, \quad v=0,\dots,N_y-1,
\end{align*}
where 
$$ \alpha_x(u) \triangleq \left\{\begin{array}{ll} \sqrt{\frac{1}{N_x}}, & u = 0 \\ 
\sqrt{\frac{2}{N_x}}, & u \neq 0
\end{array} \right., \quad 
\alpha_y(v) \triangleq \left\{\begin{array}{ll} \sqrt{\frac{1}{N_y}}, & v = 0 \\
\sqrt{\frac{2}{N_y}}, & v \neq 0.
\end{array} \right. $$

The inverse DCT is given by:
\begin{align}
\phi_d(I_x,I_y) &= \sum_{u=0}^{N_x-1} \sum_{v=0}^{N_y-1} \alpha_x(u) \alpha_y(v) C(u,v) \cos \left(\frac{(2I_x+1)\pi u}{2 N_x} \right) \cos \left(\frac{(2I_y+1)\pi v}{2 N_y} \right), \label{eqn:inverse_DCT} \\ &\quad\quad I_x=0,\dots,N_x-1, \quad I_y=0,\dots,N_y-1. \nonumber
\end{align}
It will be convenient for our purposes to rewrite \eqref{eqn:inverse_DCT} as
\begin{align}
\phi_d(I_x,I_y) &= \sum_{(u,v) \in \mathcal{U}}  \alpha_x(u) \alpha_y(v) C(u,v) \cos \left(\frac{(2I_x+1)\pi u}{2 N_x} \right) \cos \left(\frac{(2I_y+1)\pi v}{2 N_y} \right), \label{eqn:inverse_DCT_single_summation} \\ &\quad\quad I_x=0,\dots,N_x-1, \quad I_y=0,\dots,N_y-1, \nonumber
\end{align}
where 
$$\mathcal{U} \triangleq \{(u,v): u \in \{0, \dots, N_x-1\}, v \in \{0, \dots, N_y-1\} \}.$$

The most important information about the field distribution is concentrated in the low order modes, i.e. the components corresponding to $\cos \left(\frac{(2I_x+1)\pi u}{2 N_x} \right) \cos \left(\frac{(2I_y+1)\pi v}{2 N_y} \right)$ with $u$ and $v$ small, while higher order modes define the fine structure of the field distribution. See Figs. \ref{fig:field_seed341_DCT} and \ref{fig:field_seed343_DCT} for examples of how retaining different numbers of modes affects the quality of the approximation to the true field.

\section{System Model}
\label{sec:system_model}

\subsection{Field Model}
Motivated by the above discussion, we propose to approximate \eqref{eqn:inverse_DCT_single_summation} by
\begin{equation}
\label{field_model} 
\begin{split}
\phi_d(I_x,I_y) & \approx \sum_{(u,v)\in \tilde{\mathcal{U}}} \alpha_x(u) \alpha_y(v) C(u,v) \cos \left(\frac{(2I_x+1)\pi u}{2 N_x} \right) \cos \left(\frac{(2I_y+1)\pi v}{2 N_y} \right), \\ & \quad\quad I_x=0,\dots,N_x-1, \quad I_y=0,\dots,N_y-1 \\ & \triangleq \tilde{\phi}_d(I_x,I_y),
\end{split}
\end{equation}
where $\tilde{\mathcal{U}} \subseteq \mathcal{U}$ is the subset of low order modes that we wish to retain.\footnote{In general, we could in \eqref{field_model} use  coefficients $\tilde{C}(u,v)$ which are not necessarily equal to $C(u,v)$. One reason for taking the coefficients to be equal to $C(u,v)$ is given in Lemma \ref{lemma:optimal_C_DCT}.}

For example, we could  retain the first $\tilde{N}_x \times \tilde{N}_y$ modes, with $\tilde{N}_x \leq N_x, \tilde{N}_y \leq N_y$, so that 
\begin{equation}
\label{eqn:U_tilde_rect}
 \tilde{\mathcal{U}} = \{(u,v): u \in \{0, \dots, \tilde{N}_x-1\}, v \in \{0, \dots, \tilde{N}_y-1\} \}.
 \end{equation}
The total number  of modes retained $\tilde{N} $ is thus equal to $\tilde{N} = \tilde{N}_x  \tilde{N}_y$.

Another possibility is the following:
\begin{equation}
\label{eqn:U_tilde_largest}
\tilde{\mathcal{U}} = \{ \tilde{N} \textnormal{ pairs } (u,v) \textnormal{ with smallest values of } (u+1)^2 + (v+1)^2 \}
\end{equation}
which tries to retain the $\tilde{N}$ ``largest'' (in magnitude) modes.\footnote{This is of course an approximation, as exactly determining the $\tilde{N}$ largest modes depends on and requires knowledge of the very field that we are trying to estimate.} The motivation for \eqref{eqn:U_tilde_largest} comes from a result that the DCT coefficients $C(u,v)$ decay as $O \big(\frac{1}{(u+1)^2 + (v+1)^2} \big)$ for $u,v \rightarrow \infty$ \cite{YamataniSaito}. Thus the larger components will usually have smaller values of  $(u+1)^2 + (v+1)^2$, leading to the choice \eqref{eqn:U_tilde_largest}. In numerical simulations, we have found \eqref{eqn:U_tilde_largest} to give better approximations than \eqref{eqn:U_tilde_rect} (for the same number of retained modes $\tilde{N}$) in many, though not all, cases.

\subsection{Measurement Model}
In this paper we consider the following noisy quantized measurement model for a vehicle at position index $(I_x,I_y)$:
\begin{equation}
\label{quantized_measurement_model}
z(I_x,I_y) = q(\phi_d(I_x,I_y) + n(I_x, I_y))
\end{equation}
where $n(\bm{\cdot},\bm{\cdot})$ is random noise and $q(\bm{\cdot})$ is a quantizer of $L$ levels, say $\{0, 1, \dots, L-1\}$. The quantizer can be expressed in the form 
\begin{equation}
\label{eqn:quantizer}
q(x) = \left\{\begin{array}{cc} 0, & x < \tau_0 \\ 1, & \tau_0 \leq x < \tau_1 \\ \vdots & \vdots \\ L-2, & \tau_{L-3} \leq x < \tau_{L-2} \\ L-1, & x \geq \tau_{L-2}    \end{array} \right. 
\end{equation}
where the quantizer thresholds $\{\tau_0,\dots,\tau_{L-2}\}$ satisfy $\tau_0 \leq \tau_1 \leq \dots \leq \tau_{L-2}$. The use of a quantized measurement model is motivated by the fact that many chemical sensors can only provide output from a small number of discrete bars \cite{RobinsRapleyThomas,ChengKondaSinghScott}.

%At position $\big(X_{\textnormal{min}} + \left(1/2 + I_x \right) \Delta_x, Y_{\textnormal{min}} + \left(1/2 + I_y \right) \Delta_y \big)$, we have noisy measurements of the field
%$$z(I_x,I_y) = h(\phi_d(I_x,I_y), n(I_x, I_y)),$$
%where $h(\bm{\cdot},\bm{\cdot})$ is a (in general non-linear) function, and  $n(\bm{\cdot},\bm{\cdot})$ is random noise.

%For example, we could have additive noise
%\begin{equation}
%\label{additive_noise_model}
%z(I_x,I_y) = \phi_d(I_x,I_y) + n(I_x, I_y),
%\end{equation}
%similar to \cite{LaSheng,LaShengChen}. 

\begin{remark}
The special case of \eqref{quantized_measurement_model} corresponding to a 2-level quantizer, or binary measurements, is considered in \cite{LeongZamani_SP,LeongZamaniShames,TranGarratt}. 
It can be expressed as
\begin{equation}
\label{binary_measurement_model}
z(I_x,I_y) = \mathds{1env} \big(\phi_d(I_x,I_y) + n(I_x, I_y) > \tau \big),
\end{equation}
where $\tau$ is the quantizer threshold, and $\mathds{1}(\bm{\cdot})$ is the indicator function that returns 1 if its argument is true and 0 otherwise. 
Other measurement models which have been considered in the literature include additive noise models \cite{LaSheng,LaShengChen} and Poisson measurement models \cite{MorelandeSkvortsov}.    
\end{remark}

%Another measurement model which has been considered are Poisson measurements \cite{MorelandeSkvortsov}. Define $\mathbf{x} \triangleq (x,y)$. Then in this model
%$$z(\mathbf{x}) \sim \texttt{Poisson}(\lambda(\textbf{x})),$$
%where 
%$$\lambda(\textbf{x}) = \int k(\mathbf{x}' - \mathbf{x}) \phi(\mathbf{x}') d\mathbf{x}'$$
%and 
%$$k (\mathbf{x}) = \left\{ \begin{array}{cc} \frac{1}{R^2}, & ||\mathbf{x}|| \leq R \\  \frac{1}{||\mathbf{x}||^2}, & ||\mathbf{x}|| \geq R \end{array} \right.$$
%for some constant $R$.

\subsection{Problem Statement}
\label{sec:problem_statement}
The problem we wish to consider in this paper is to estimate the coefficients\footnote{When we refer to \emph{estimation of components/modes} in this paper, we specifically mean estimation of the coefficients $C(u,v)$.}
$$C(u,v), \,\,(u,v) \in \tilde{\mathcal{U}}$$
of the field $\phi_d(I_x,I_y)$, from quantized measurements $\{z(I_x,I_y)\}$ collected by an unmanned autonomous vehicle under the measurement model \eqref{quantized_measurement_model}. The estimation should be done in an online manner such that the estimates are continually updated as new measurements are collected.

\section{Comparison with RBF Field Model}
\label{sec:DCT_RBF_comparison}
Before we consider the problem of estimating the coefficients $C(u,v)$ (which will be studied in Section \ref{sec:DCT_estimation}), we will in this section compare the use of our Fourier component model
\eqref{field_model}
with the radial basis function model considered in \cite{RazakSukumarChung_journal, LeongZamani_SP,LeongZamaniShames,TranGarratt} (see also \cite{LaSheng,LaShengChen,MorelandeSkvortsov} for similar models), in terms of how well they can approximate a field for a given number of modes (for the Fourier component model) or basis functions (for the RBF model). 

\subsection{Fourier Component Field Model}
Define the mean squared error (MSE):
\begin{equation}
\label{eqn:MSE_Fourier}
 \textnormal{MSE} \triangleq \frac{1}{N_x N_y} \sum_{I_x=0}^{N_x-1} \sum_{I_y=0}^{N_y-1} \Big( \phi_d (I_x, I_y) - \tilde{\phi}_d (I_x, I_y)  \Big)^2,
 \end{equation}
where 
%$$\phi_d(I_x,I_y) = \sum_{(u,v) \in \mathcal{U}} \alpha_x(u) \alpha_y(v) C(u,v) \cos \left(\frac{(2I_x+1)\pi u}{2 N_x} \right) \cos \left(\frac{(2I_y+1)\pi v}{2 N_y} \right)$$
 $\phi_d(I_x,I_y)$ is the (discretized) true field given by \eqref{eqn:inverse_DCT_single_summation} and
$$\tilde{\phi}_d (I_x, I_y)   \triangleq \sum_{(u,v) \in \tilde{\mathcal{U}} } \alpha_x(u) \alpha_y(v) \tilde{C}(u,v) \cos \left(\frac{(2I_x+1)\pi u}{2 N_x} \right) \cos \left(\frac{(2I_y+1)\pi v}{2 N_y} \right)$$
is the approximation of the true field using a subset of modes $\tilde{U}$ and coefficients $\tilde{C}(u,v)$. The expression for $\tilde{\phi}_d (I_x, I_y) $ is the same as \eqref{field_model} except that the coefficients $\tilde{C}(u,v)$ may be different from $C(u,v)$. However, it turns out that setting $\tilde{C}(u,v)$ to be equal to $C(u,v)$ will minimize the MSE. 

\begin{lemma}
\label{lemma:optimal_C_DCT}
Given a subset of modes $\tilde{U}$, the optimal values of $\tilde{C}(u,v)$ that minimize \eqref{eqn:MSE_Fourier} satisfy
\begin{equation}
\label{eqn:optimal_C}
 \tilde{C}^*(u,v) = C(u,v), \, \forall (u,v) \in \tilde{U}.
 \end{equation}
\end{lemma}
\begin{proof}
See the Appendix.
\end{proof}

\begin{remark}
One can also minimize \eqref{eqn:MSE_Fourier} by treating it as a linear least squares problem \cite{CalafioreElGhaoui,Murphy_book1}, which will numerically give the same solution, however the analytical expression \eqref{eqn:optimal_C} provided by Lemma \ref{lemma:optimal_C_DCT} is much more explicit. 
\end{remark}

\subsection{RBF Field Model}
The following RBF field model is used in \cite{RazakSukumarChung_journal, LeongZamani_SP,LeongZamaniShames,TranGarratt}:
\begin{equation}
\label{field_model_RBF}
\phi(\mathbf{x}) \approx \sum_{j=1}^J \beta_j K_j(\mathbf{x}),
\end{equation}
where $\mathbf{x} \triangleq (x,y)$ and $K_j(\mathbf{x}), j=1,\dots,J$ are radial basis functions. In particular, we consider the choice
 \begin{equation}
 \label{eqn:Gaussian_RBF}
 K_j(\mathbf{x}) = \exp \left(- \frac{\|\mathbf{c}_j-\mathbf{x}\|^2}{\sigma_j^2}\right), \quad j=1,\dots,J,
 \end{equation}
which results in a Gaussian mixture model \cite{MorelandeSkvortsov}.
For a given number of basis functions $J$, we assume that the $\mathbf{c}_j$'s and $\sigma_j$'s are chosen,\footnote{The case where the $\mathbf{c}_j$'s and $\sigma_j$'s are also estimated has been considered, but was found to suffer from identifiability issues and sometimes give very unreliable results \cite{LeongZamani_SP}.} while the $\beta_j$'s are free parameters. Algorithms for estimating the $\beta_j$'s are studied in, e.g., \cite{RazakSukumarChung_journal, LeongZamani_SP,LeongZamaniShames,TranGarratt}. Here we consider instead the problem of finding the optimal  $\beta_j$'s in order to minimize the mean squared error, to see how good the RBF model can be when approximating a field for a given set of basis functions. Define
\begin{equation}
\label{eqn:MSE_RBF}
\textnormal{MSE}_{RBF} \triangleq \frac{1}{|\mathcal{S}_d|} \sum_{\mathbf{x} \in \mathcal{S}_d} \Big( \phi (\mathbf{x}) - \sum_{j=1}^J \beta_j K_j(\mathbf{x}) \Big)^2,    
\end{equation}
 where $\phi (\mathbf{x}) $ is the true field value at position~$\mathbf{x}$, $\mathcal{S}_d$ is a discretized set of points in the search region $\mathcal{S}$, and $|\mathcal{S}_d|$ is the cardinality of $\mathcal{S}_d$. 

\begin{lemma}
\label{lemma:optimal_beta_RBF}
Given a set of radial basis functions $\{K_1(.), \dots, K_j(.)\}$ and an ordering $\{\mathbf{x}_1, \dots, \mathbf{x}_{|\mathcal{S}_d|} \}$ of the elements in $\mathcal{S}_d$, the optimal values of $(\beta_1, \dots, \beta_J)$ that minimize \eqref{eqn:MSE_RBF} satisfy
$$ \bm{\beta}^* = \left( \mathcal{K}^T \mathcal{K} \right)^{-1} \mathcal{K}^T \bm{\phi},$$
where $\bm{\beta} = \begin{bmatrix} \beta_1 & \dots & \beta_J \end{bmatrix}^T$, $\bm{\phi} = \begin{bmatrix} \phi(\mathbf{x}_1), \dots, \phi(\mathbf{x}_{|\mathcal{S}_d|}) \end{bmatrix}^T$, and 
$$\mathcal{K} = \begin{bmatrix}
K_1(\mathbf{x}_1) & \dots & K_J(\mathbf{x}_1) \\
\vdots & \ddots & \vdots \\
K_1(\mathbf{x}_{|\mathcal{S}_d|}) & \dots & K_J(\mathbf{x}_{|\mathcal{S}_d|})
\end{bmatrix}.$$
\end{lemma}

\begin{proof}
This is a standard application of the optimal solution to a linear least squares / linear regression problem \cite{CalafioreElGhaoui,Murphy_book1}.    
\end{proof}

\subsection{Numerical Experiments}

\begin{figure}[t!]
\centering 
\includegraphics[scale=0.35]{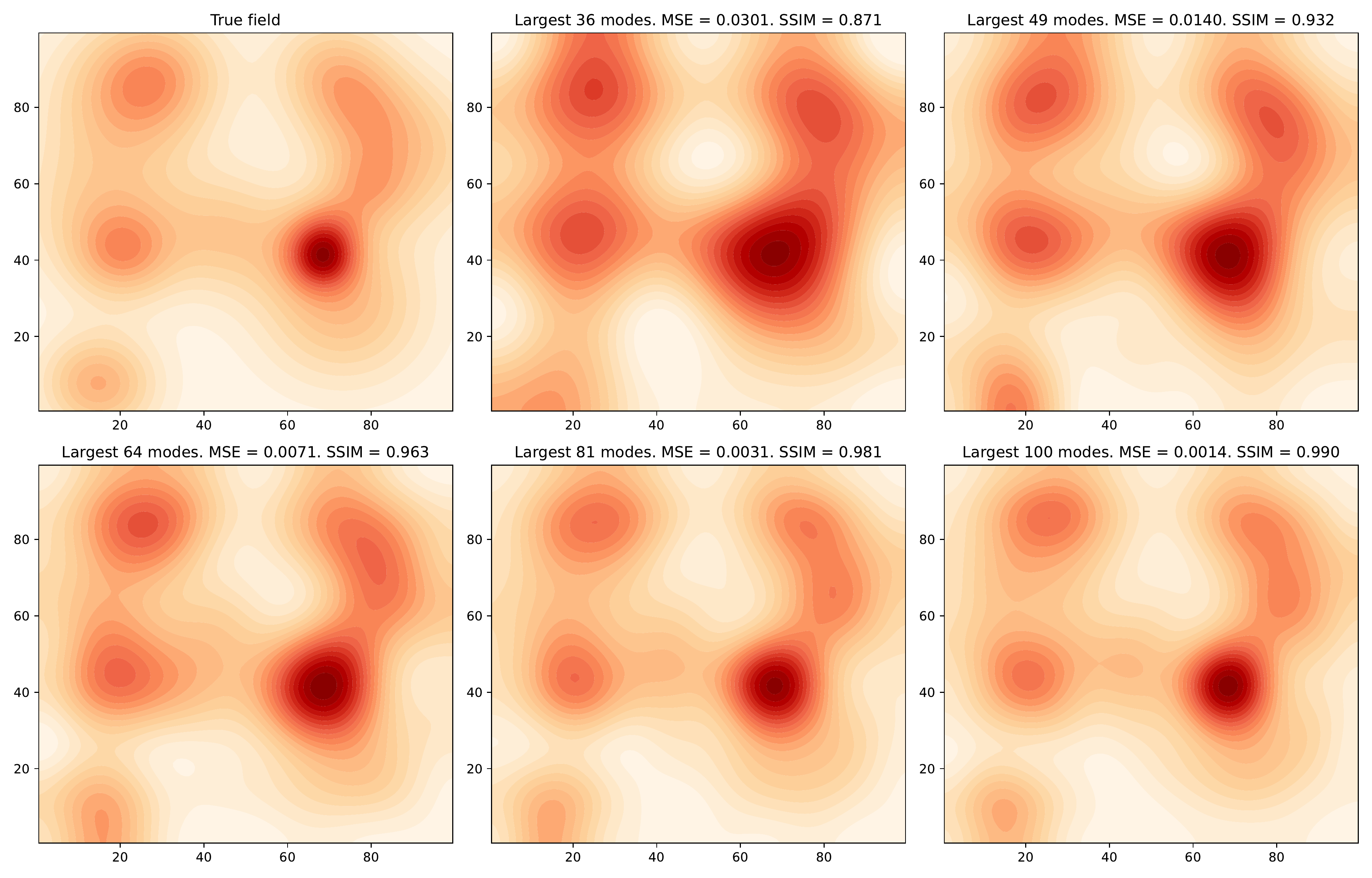} 
\caption{True field and approximations obtained by retaining different numbers of modes}
\label{fig:field_seed341_DCT}
\end{figure} 

\begin{figure}[t!]
\centering 
\includegraphics[scale=0.35]{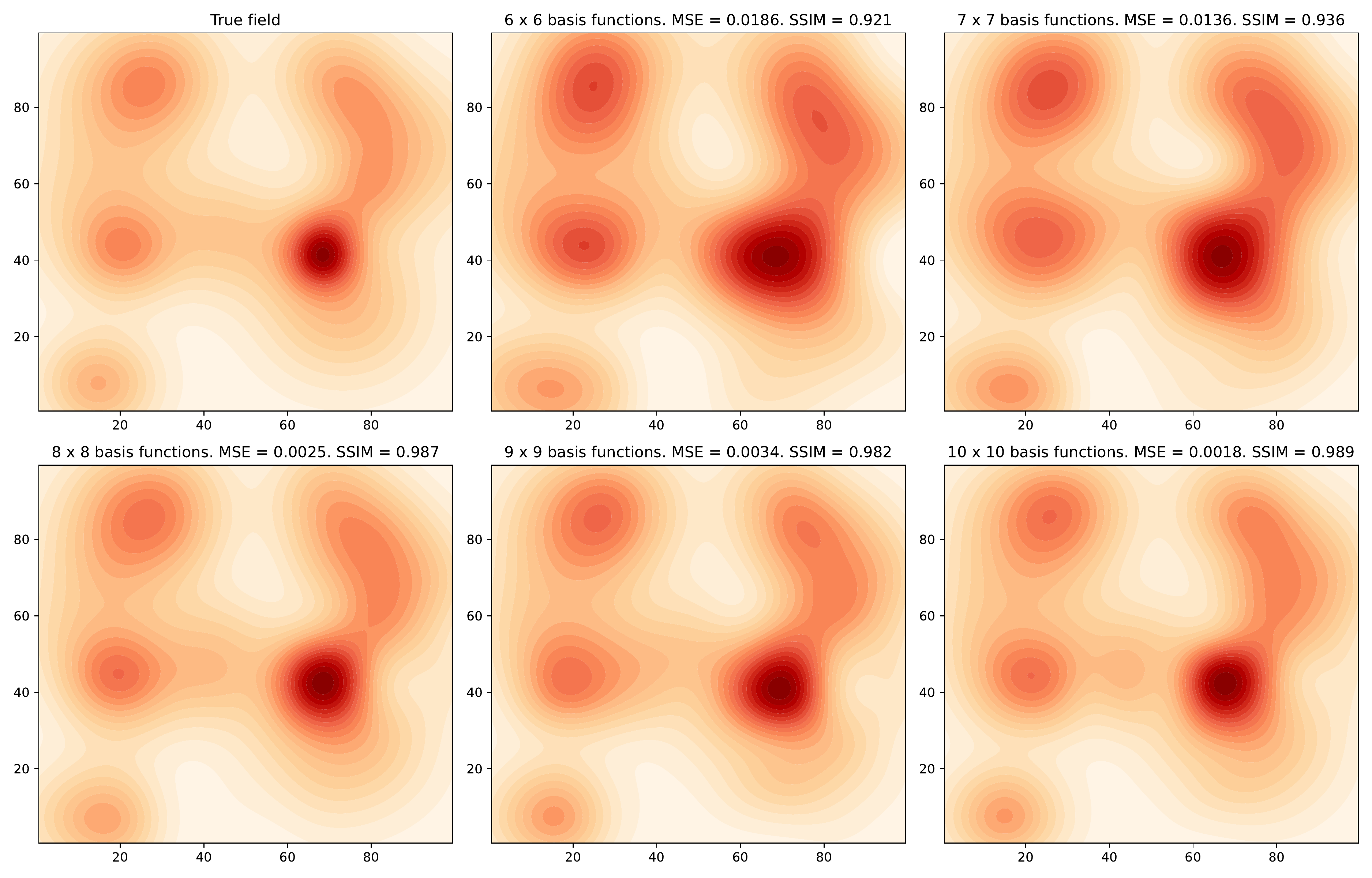} 
\caption{True field and approximations obtained by using different numbers of basis functions}
\label{fig:field_seed341_RBF}
\end{figure} 

\begin{figure}[t!]
\centering 
\includegraphics[scale=0.35]{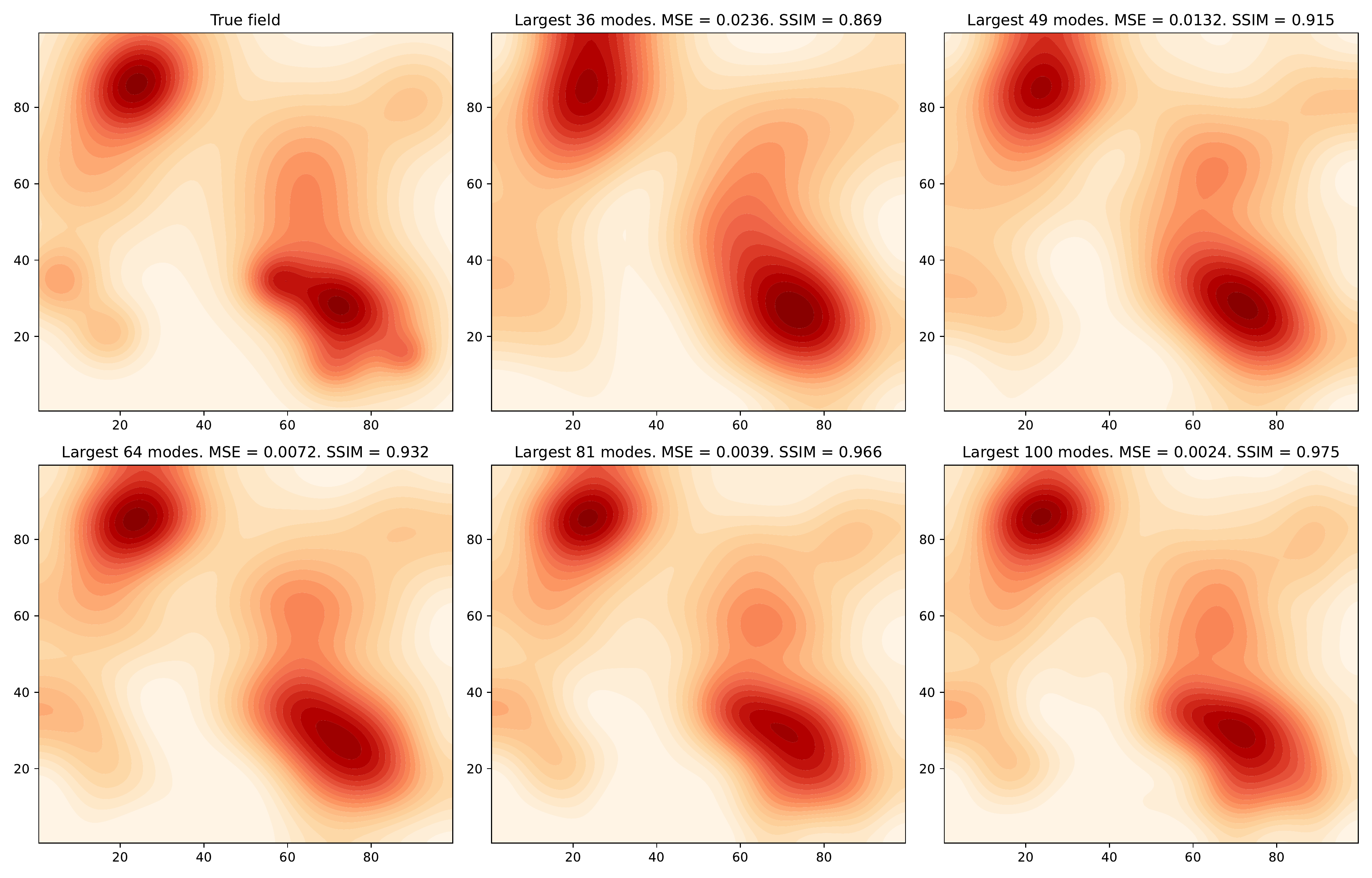} 
\caption{True field and approximations obtained by retaining different numbers of modes}
\label{fig:field_seed343_DCT}
\end{figure} 

\begin{figure}[t!]
\centering 
\includegraphics[scale=0.35]{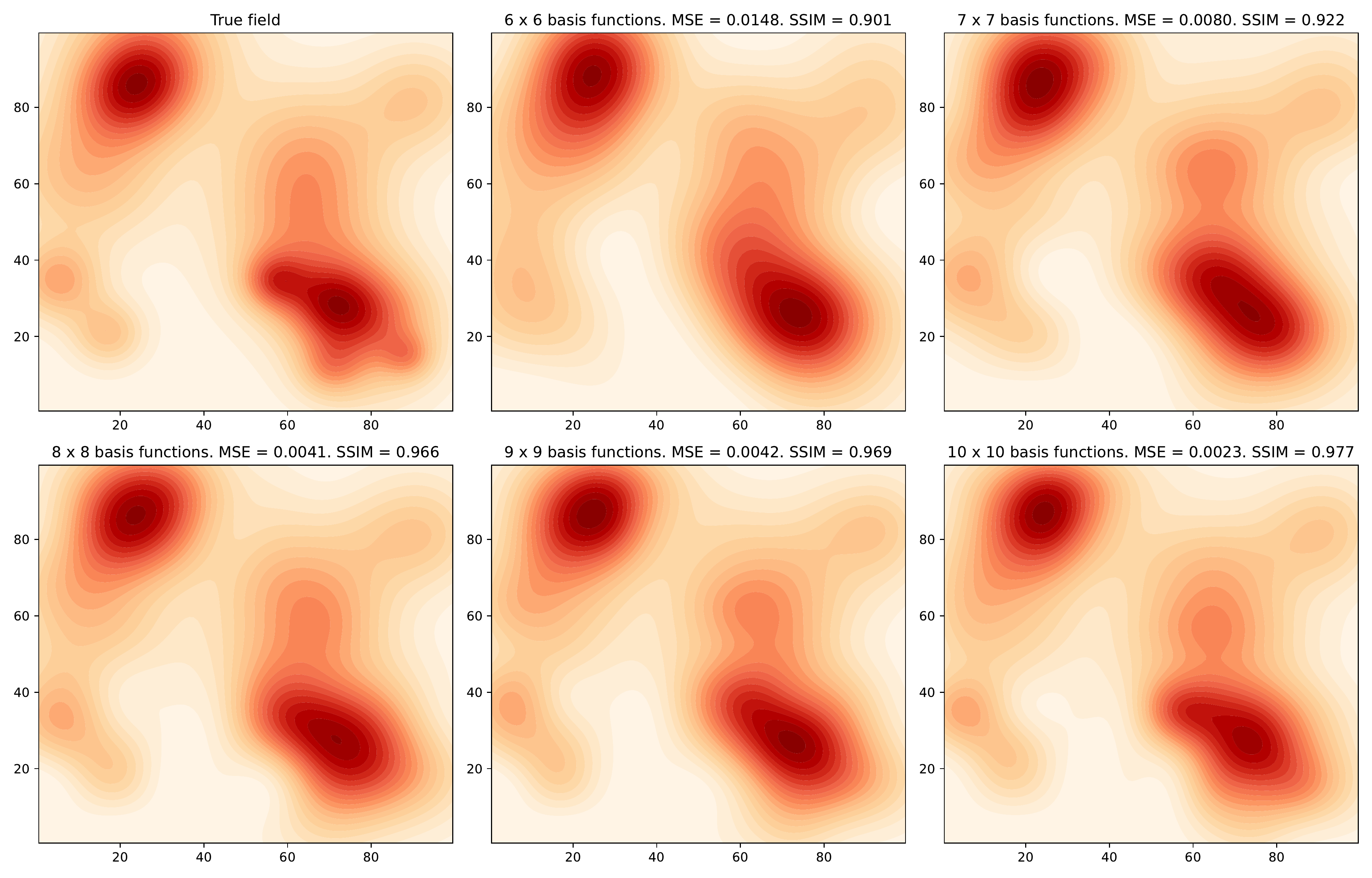} 
\caption{True field and approximations obtained by using different numbers of basis functions}
\label{fig:field_seed343_RBF}
\end{figure} 

In Figs.\ref{fig:field_seed341_DCT}-\ref{fig:field_seed343_RBF}  we show two example fields, and the field approximations that are obtained when various different numbers of modes (for Fourier component model) or radial basis functions (for RBF model) are used. The discretization in the true fields is set as $N_x = N_y = 100$ (so that there are $100^2 = 10000$ modes in total). For the RBF model we set $\mathcal{S}_d = \mathcal{X}_d \times \mathcal{Y}_d$, so that the discretized set of points are the same in the MSE calculations. 
For the Fourier component model, we use the optimal choice of $\tilde{C}^*(u,v)$ given in Lemma \ref{lemma:optimal_C_DCT} (which corresponds to the model \eqref{field_model}), and we choose $\tilde{\mathcal{U}}$ as in \eqref{eqn:U_tilde_largest} to retain the $\tilde{N}$ ``largest'' modes. For the RBF model, we use $J = J_x \times J_y$ radial basis functions with $\mathbf{c}_j$'s in \eqref{eqn:Gaussian_RBF} placed uniformly on a grid at locations $\mathcal{X}_{RBF} \times \mathcal{Y}_{RBF}$, where
\begin{align*}
\mathcal{X}_{RBF} \triangleq \left\{X_{\textnormal{min}} + \Big(\frac{1}{2} + i_x \Big) \delta_x:   i_x \in \{0, \dots, J_x - 1\} \right\} \\
\mathcal{Y}_{RBF} \triangleq \left\{Y_{\textnormal{min}} + \Big(\frac{1}{2} + i_y \Big) \delta_y:  i_y \in \{0, \dots, J_y - 1\} \right\}
\end{align*}
and 
$$\delta_x \triangleq \frac{X_{\textnormal{max}} - X_{\textnormal{min}}}{J_x}, \quad \delta_y \triangleq \frac{Y_{\textnormal{max}} - Y_{\textnormal{min}}}{J_y}. $$
The $\sigma_j$'s in \eqref{eqn:Gaussian_RBF} are chosen to be equal to $\sigma_j = \max(\delta_x,\delta_y), \forall j$. The $\beta_j$'s used in \eqref{field_model_RBF} are the optimal values computed according to Lemma \ref{lemma:optimal_beta_RBF}. 

In the figures we show two performance measures, 1) the MSE, and 2) the structural similarity (SSIM) index, which originated in \cite{WangBovikSheikhSimoncelli} and has been widely adopted in the image processing community. The structural similarity index is a measure of the similarity between two images. In our case, we can regard $\Phi = \{\phi_d(I_x, I_y): I_x = 0,\dots, N_x - 1, I_y = 0, \dots, N_y - 1 \}$ and $\tilde{\Phi} = \{\tilde{\phi}_d(I_x, I_y): I_x = 0,\dots, N_x - 1, I_y = 0, \dots, N_y - 1 \}$ as the image representations of the true and approximated fields respectively, and compute the SSIM between these two images. The SSIM gives a scalar value between 0 and 1, with $\textnormal{SSIM} = 1$ if the two images to be compared are identical. 
We refer to \cite{WangBovikSheikhSimoncelli,WangBovik_MSE} for the specific equations used to compute the SSIM. 

We see from Figs.\ref{fig:field_seed341_DCT}-\ref{fig:field_seed343_RBF} that as more modes (for Fourier component model) or basis functions (for RBF model) are used, the approximations to the true field improves.  When using a smaller number of modes / basis functions the RBF model seems to give better approximations than the Fourier component model, while for larger numbers of modes / basis functions the two approaches perform similarly. We also observe that for these two examples, using a relatively small number of modes (when compared to the total number of modes of 10000) or basis functions will still result in a qualitatively reasonable approximation to the true field.

While the Fourier component model does not seem to offer a significant advantage in terms of approximation quality, there are other reasons where one may consider its use. One advantage of the Fourier model is that it provides a natural way to control the number of model parameters (the coefficients $C(u,v)$) that need to be estimated, in that we simply choose however many modes we wish to retain, whereas with the RBF model one would need to also choose the locations $\mathbf{c}_j$ to place the basis functions and what the values of $\sigma_j$ should be. Additionally, if we want to refine our field model with finer structure by including more model parameters, in the Fourier component model we can reuse any previous estimates  (and further improve them) of the lower order modes (see Section \ref{sec:field_refinement}), since these remain the same in a model with more modes, whereas in the RBF model one would likely need  to recalculate the estimates of all the parameter values when more basis functions are utilized.

\section{Estimation of Fourier Components}
\label{sec:DCT_estimation}
We now return to the problem of estimating the coefficients $C(u,v), \, (u,v) \in \tilde{\mathcal{U}}$ stated in Section~\ref{sec:problem_statement}. 
Given a set of modes to be retained $\tilde{\mathcal{U}}$, of cardinality $\tilde{N}$, define an ordering on $\tilde{\mathcal{U}}$ indexed by $j \in \{0, \dots, \tilde{N}-1\}$. For instance, the elements of $\tilde{\mathcal{U}}$ could be sorted in lexicographic order.  
Denote the $j$-th element under this ordering  by $(u_j, v_j)$, and define 
$$ C_j \triangleq C(u_j,v_j).$$

%Given $u \in \{0,\dots,\tilde{N}_x-1\}$ and $v \in \{0, \dots, \tilde{N}_y-1\}$, define Alternatively, one can define $j \triangleq v + u \tilde{N}_y$, and conversely $u_j  \triangleq j \textnormal{ mod } \tilde{N}_y $, $v_j \triangleq \lfloor j/\tilde{N}_y \rfloor$. 
%$$j \triangleq u + v \tilde{N}_x$$
%and note that $j \in \{0,\dots, \tilde{N}_x  \tilde{N}_y-1\}$. Conversely, given  $j \in \{0,\dots, \tilde{N}_x  \tilde{N}_y-1\}$, define
%$$ u_j \triangleq  \lfloor j/\tilde{N}_x \rfloor, \quad v_j \triangleq j \textnormal{ mod } \tilde{N}_x$$ 
%Now define
%$$ C_j \triangleq C(u_j,v_j).$$
%Forming $C_j = C(u_j, v_j), j = 0, \dots, \tilde{N}_x \tilde{N}_y - 1$ then corresponds to applying the $\mathtt{vec}$ operation \cite{HornJohnson2} on the matrix with entries $C(u,v)$.

Then we can express
\begin{align*}
\tilde{\phi}_d(I_x,I_y) & = \sum_{(u,v) \in \tilde{\mathcal{U}}} \alpha_x(u) \alpha_y(v) C(u,v) \cos \left(\frac{(2I_x+1)\pi u}{2 N_x} \right) \cos \left(\frac{(2I_y+1)\pi v}{2 N_y} \right)
\end{align*}
in the alternative form 
\begin{equation}
\label{field_model_vector}
 \tilde{\phi}_d(I_x,I_y)  = \sum_{j=0}^{\tilde{N}-1} \alpha_x(u_j) \alpha_y(v_j) C_j \cos \left(\frac{(2I_x+1)\pi u_j}{2 N_x} \right) \cos \left(\frac{(2I_y+1)\pi v_j}{2 N_y} \right),
\end{equation}
which is a linear function of $(C_0, \dots, C_{\tilde{N} -1})$.

\begin{remark}
\label{remark:param_magnitudes}
The DCT coefficients which we are trying to estimate can be of substantially different orders of magnitude, with the higher order components being much smaller in magnitude than the ``DC'' component corresponding to  $u=v=0$, due to the result that the DCT coefficients decay as $O \big(\frac{1}{(u+1)^2 + (v+1)^2} \big)$ \cite{YamataniSaito}.
%For instance, the field in Fig. \ref{fig:field_1_modes} has $C(0,0) = 56.31$, $C(1,1) = 7.749$, $C(2,2) = -2.065$, $C(3,3) = 1.040$, $C(4,4) = 0.544$, $C(5,5) = 0.0156$, $C(6,6) = 0.0284$, $C(7,7) = 0.0075$, \dots. 
In order to estimate parameters with such large differences in magnitude, it is desirable to appropriately scale the parameters that are to be estimated, see \eqref{eqn:param_scaling} below. 
\end{remark}

\begin{remark}
Comparing \eqref{field_model_vector} with the RBF field model \eqref{field_model_RBF}, we see that they are both linear functions of the parameters that are to be estimated. Thus the algorithms developed in e.g. \cite{LaSheng,LaShengChen,RazakSukumarChung_journal, LeongZamani_SP,LeongZamaniShames,TranGarratt,MorelandeSkvortsov}  for estimating fields can in principle also be adapted to work for the field model \eqref{field_model_vector}, under their various assumed measurement models.
\end{remark}

\subsection{Estimation of Fourier Components Using Quantized Measurements}
\label{sec:DCT_estimation_quantized_measurements}
In this subsection we describe an approach to estimating the parameters $C(u,v), \, u=0, \dots, \tilde{N} - 1$, which assumes the quantized measurement model \eqref{quantized_measurement_model}-\eqref{eqn:quantizer}, with the parameters estimated recursively. 
%A sequential Monte Carlo (SMC) approach is presented in Section \ref{sec:SMC_approach}, while an online optimization approach  is presented in Section \ref{sec:online_optim_approach}. The algorithms are similar to those of \cite{LeongZamani_SP} and \cite{LeongZamaniShames} respectively,  
The algorithm uses an online optimization approach similar to \cite{LeongZamaniShames},
however in this paper we will generalize \cite{LeongZamaniShames} from binary measurements to multi-level quantized measurements, and also extend the approach to handle time-varying fields. 

For the measurement model \eqref{quantized_measurement_model}-\eqref{eqn:quantizer}, $n(\bm{\cdot},\bm{\cdot})$ is taken as zero mean noise (not necessarily Gaussian). 
Recalling the observation in Remark \ref{remark:param_magnitudes}, we consider the following scaling of the DCT coefficients:
\begin{equation}
\label{eqn:param_scaling}
\beta_j \triangleq \big((u_j+1)^2 + (v_j+1)^2 \big) C_j,
\end{equation}
and define 
$$\bm{\beta} \triangleq (\beta_0, \dots, \beta_{\tilde{N}-1})$$
as the vector of parameters that are to be estimated. 

We first introduce some notation. 
Let $z_k$ denote the measurement,  and $(I_{x,k}, I_{y,k})$ the position index, at time/iteration $k$. For notational compactness we also denote
\begin{equation}
\label{eqn:I_x_vector}
\bm{I}_{\textbf{x},k} \triangleq (I_{x,k}, I_{y,k})
\end{equation}
and  
\begin{equation}
\label{eqn:K_vector}
\mathbf{K}(\bm{I}_{\textbf{x},k} ) \triangleq \left[\begin{array}{cccc} K_0 (\bm{I}_{\textbf{x},k} ) & K_1 (\bm{I}_{\textbf{x},k} ) &  \dots & K_{\tilde{N}-1}(\bm{I}_{\textbf{x},k} ) \end{array} \right]^T,
\end{equation}
where
\begin{equation}
\label{eqn:K_vector_components}
K_j (\bm{I}_{\textbf{x},k} ) \triangleq \frac{\alpha_x(u_j) \alpha_y(v_j) }{(u_j\!+\!1)^2 \!+\! (v_j\!+\!1)^2} \cos \Big(\frac{(2I_{x,k}+1)\pi u_j}{2 N_x} \Big) \cos \Big(\frac{(2I_{y,k}+1)\pi v_j}{2 N_y} \Big).
\end{equation}
Denote 
\begin{equation}
\label{eqn:z_all}
z_{1:k} \triangleq \{z_1, \dots, z_k\}    
\end{equation}
as the set of measurements collected up to time $k$, with corresponding position indices
\begin{equation}
\label{eqn:I_x_all}
\bm{I}_{\textbf{x},1:k} \triangleq \{(I_{x,1:k}, I_{y,1:k})\}.
\end{equation}

The idea is to recursively estimate $\bm{\beta}$ by trying to minimize a cost function 
\begin{equation}
\label{eqn:cost_fn}
J_k (\bm{\beta}; \bm{I}_{\textbf{x},1:k}, z_{1:k}) = \sum_{t=0}^k g_t (\bm{\beta}; \bm{I}_{\textbf{x},t}, z_t)
\end{equation}
using online optimization techniques \cite{LesageLandryTaylorShames}.
For binary measurements \eqref{binary_measurement_model}, the following per stage cost function from \cite{LeongZamaniShames} can be used:
\begin{equation}
\label{eqn:per_stage_cost_binary}
g_t (\bm{\beta}; \bm{I}_{\textbf{x}}, z) = \left\{\begin{array}{cc} \log(1+\exp(\eta (\bm{\beta}^T \mathbf{K}(\bm{I}_{\textbf{x}}) - \tau))), & z = 0 \\ 
\log(1+\exp(-\eta (\bm{\beta}^T \mathbf{K}(\bm{I}_{\textbf{x}}) - \tau))), & z = 1 
\end{array}
\right.
\end{equation}
where $\eta>0$ is a parameter in the logistic function $\ell(x) \triangleq 1/(1+\exp(\eta x))$, where larger values of $\eta$ will more closely approximate the function $\mathds{1env}(x > 0)$. The cost function \eqref{eqn:per_stage_cost_binary} is similar to cost functions used in binary logistic regression problems \cite[p.516]{CalafioreElGhaoui}. In the current work, we wish to define a cost suitable for multi-level quantized measurements. Note that there are cost functions used in multinomial logistic regression problems \cite{Murphy_book1}, however they are unsuitable for our problem as they usually involve multiple sets of parameters for each possible output $z$, whereas here we just have a single set of parameters $\bm{\beta}$. 

To motivate our cost function, let us look more closely at the binary measurements cost function \eqref{eqn:per_stage_cost_binary}. In the case where the measurement  $z_t$ at time $t$ and position index $\bm{I}_{\textbf{x},t}$ is equal to 0, the cost $g_t (\bm{\beta}; \bm{I}_{\textbf{x},t}, z_t)$ will be small if  $\bm{\beta}^T \mathbf{K}(\bm{I}_{\textbf{x},t}) $ is less than the quantizer threshold $\tau$, and large otherwise. Similarly, when  $z_t=1$,  $g_t (\bm{\beta}; \bm{I}_{\textbf{x},t}, z_t)$ will be small if  $\bm{\beta}^T \mathbf{K}(\bm{I}_{\textbf{x},t}) $ is greater than $\tau$, and large otherwise. For the case of multi-level quantized measurements  with $L$ levels  given by \eqref{eqn:quantizer}, we would like to have a cost function such that 1) when $z_t=0$, $g_t (\bm{\beta}; \bm{I}_{\textbf{x},t}, z_t)$ is small for $\bm{\beta}^T \mathbf{K}(\bm{I}_{\textbf{x},t}) < \tau_0$, and large otherwise, 2)  when $z_t=l, \, l \in \{1,\dots,L-2\}$, $g_t (\bm{\beta}; \bm{I}_{\textbf{x},t}, z_t)$ is small for $\tau_{l-1} \leq \bm{\beta}^T \mathbf{K}(\bm{I}_{\textbf{x},t}) < \tau_l$, and large otherwise, and 3) when $z_t=L-1$,  $g_t (\bm{\beta}; \bm{I}_{\textbf{x},t}, z_t)$ is small for $\bm{\beta}^T \mathbf{K}(\bm{I}_{\textbf{x},t}) > \tau_{L-2}$, and large otherwise. In this paper we will choose the following per stage cost function, which can be easily checked to satisfy these three requirements:
\begin{equation}
\label{eqn:per_stage_cost_multilevel}
g_t (\bm{\beta}; \bm{I}_{\textbf{x}}, z) \triangleq \left\{\begin{array}{ll} \log(1+\exp(\eta (\bm{\beta}^T \mathbf{K}(\bm{I}_{\textbf{x}}) - \tau_0))), & z = 0 \\ 
\log(1+\exp(-\eta (\bm{\beta}^T \mathbf{K}(\bm{I}_{\textbf{x}}) - \tau_{z-1}))) & \\
 \quad + \log(1+\exp(\eta (\bm{\beta}^T \mathbf{K}(\bm{I}_{\textbf{x}}) - \tau_{z}))), & z  \in \{1,\dots,L-2\}\\
\log(1+\exp(-\eta (\bm{\beta}^T \mathbf{K}(\bm{I}_{\textbf{x}}) - \tau_{L-2}))), & z = L-1. 
\end{array}
\right.
\end{equation}
We remark that \eqref{eqn:per_stage_cost_multilevel} reduces to \eqref{eqn:per_stage_cost_binary} when the measurements are binary. 

Now that the per stage cost \eqref{eqn:per_stage_cost_multilevel} has been defined, we will present the online estimation algorithm. First, the gradient of $g_t(\bm{\cdot};\bm{\cdot},\bm{\cdot})$ can be derived as 
\begin{equation}
\label{eqn:per_stage_gradient}
\nabla g_t (\bm{\beta}; \bm{I}_{\textbf{x}}, z) = \left\{\begin{array}{ll} 
\frac{\eta}{1+\exp(-\eta (\bm{\beta}^T \mathbf{K}(\bm{I}_{\textbf{x}}) - \tau_0))} \mathbf{K}(\bm{I}_{\textbf{x}}), & z = 0 \\
\Big( \frac{-\eta}{1+\exp(\eta (\bm{\beta}^T \mathbf{K}(\bm{I}_{\textbf{x}}) - \tau_{z-1}))} & \\ \quad \quad + \frac{\eta}{1+\exp(-\eta (\bm{\beta}^T \mathbf{K}(\bm{I}_{\textbf{x}}) - \tau_z))} \Big) \mathbf{K}(\bm{I}_{\textbf{x}}), & z  \in \{1,\dots,L-2\}\\
 \frac{-\eta}{1+\exp(\eta (\bm{\beta}^T \mathbf{K}(\bm{I}_{\textbf{x}}) - \tau_{L-2}))} \mathbf{K}(\bm{I}_{\textbf{x}}), & z = L-1, 
\end{array}
\right.
\end{equation}
while the Hessian of $g_t(\bm{\cdot};\bm{\cdot},\bm{\cdot})$ can be derived as 
\begin{align}
& \nabla^2 g_t (\bm{\beta}; \bm{I}_{\textbf{x}}, z) \nonumber \\& = \left\{\begin{array}{ll} \frac{\eta^2 \exp(-\eta (\bm{\beta}^T \mathbf{K}(\bm{I}_{\textbf{x}}) - \tau_0))}{(1+\exp(-\eta (\bm{\beta}^T \mathbf{K}(\bm{I}_{\textbf{x}}) - \tau_0)))^2} \mathbf{K}(\bm{I}_{\textbf{x}}) \mathbf{K}(\bm{I}_{\textbf{x}})^T, & z = 0 \\
\Big( \frac{\eta^2 \exp(\eta (\bm{\beta}^T \mathbf{K}(\bm{I}_{\textbf{x}}) - \tau_{z-1}))}{(1+\exp(\eta (\bm{\beta}^T \mathbf{K}(\bm{I}_{\textbf{x}}) - \tau_{z-1})))^2} & \\ \quad \quad + \frac{\eta^2 \exp(-\eta (\bm{\beta}^T \mathbf{K}(\bm{I}_{\textbf{x}}) - \tau_z))}{(1+\exp(-\eta (\bm{\beta}^T \mathbf{K}(\bm{I}_{\textbf{x}}) - \tau_z)))^2} \Big)  \mathbf{K}(\bm{I}_{\textbf{x}}) \mathbf{K}(\bm{I}_{\textbf{x}})^T, & z \in \{1,\dots,L-2\}\\
 \frac{\eta^2 \exp(\eta (\bm{\beta}^T \mathbf{K}(\bm{I}_{\textbf{x}}) - \tau_{L-2}))}{(1+\exp(\eta (\bm{\beta}^T \mathbf{K}(\bm{I}_{\textbf{x}}) - \tau_{L-2})))^2} \mathbf{K}(\bm{I}_{\textbf{x}}) \mathbf{K}(\bm{I}_{\textbf{x}})^T, & z = L-1. 
\end{array}
\right. \nonumber \\
& = \left\{\begin{array}{ll} 
\frac{\eta^2 \exp(\eta (\bm{\beta}^T \mathbf{K}(\bm{I}_{\textbf{x}}) - \tau_0))}{(1+\exp(\eta (\bm{\beta}^T \mathbf{K}(\bm{I}_{\textbf{x}}) - \tau_0)))^2} \mathbf{K}(\bm{I}_{\textbf{x}}) \mathbf{K}(\bm{I}_{\textbf{x}})^T, & z = 0 \\
\Big( \frac{\eta^2 \exp(\eta (\bm{\beta}^T \mathbf{K}(\bm{I}_{\textbf{x}}) - \tau_{z-1}))}{(1+\exp(\eta (\bm{\beta}^T \mathbf{K}(\bm{I}_{\textbf{x}}) - \tau_{z-1})))^2} & \\ \quad \quad + \frac{\eta^2 \exp(\eta (\bm{\beta}^T \mathbf{K}(\bm{I}_{\textbf{x}}) - \tau_z))}{(1+\exp(\eta (\bm{\beta}^T \mathbf{K}(\bm{I}_{\textbf{x}}) - \tau_z)))^2} \Big)  \mathbf{K}(\bm{I}_{\textbf{x}}) \mathbf{K}(\bm{I}_{\textbf{x}})^T, & z \in \{1,\dots,L-2\}\\
 \frac{\eta^2 \exp(\eta (\bm{\beta}^T \mathbf{K}(\bm{I}_{\textbf{x}}) - \tau_{L-2}))}{(1+\exp(\eta (\bm{\beta}^T \mathbf{K}(\bm{I}_{\textbf{x}}) - \tau_{L-2})))^2} \mathbf{K}(\bm{I}_{\textbf{x}}) \mathbf{K}(\bm{I}_{\textbf{x}})^T, & z = L-1. 
\end{array} \right.  \label{eqn:per_stage_Hessian}
\end{align}
An approximate online Newton method \cite{LeongZamaniShames} for estimating the parameters $\bm{\beta}$ is now given by:
\begin{equation}
\label{eqn:approx_ONM1}
 \hat{\bm{\beta}}_{k+1} = \hat{\bm{\beta}}_k - \left( H_k (\hat{\bm{\beta}}_k; \bm{I}_{\textbf{x},1:k}, z_{1:k})  \right)^{-1} G_k (\hat{\bm{\beta}}_k; \bm{I}_{\textbf{x},k}, z_k), 
\end{equation}
where
\begin{align}
G_k (\hat{\bm{\beta}}_k; \bm{I}_{\textbf{x},k}, z_k) & = \nabla g_k (\hat{\bm{\beta}}_k; \bm{I}_{\textbf{x},k}, z_k)  \nonumber \\
H_k(\hat{\bm{\beta}}_{k}; \bm{I}_{\textbf{x},1:k}, z_{1:k}) & =  H_{k-1}(\hat{\bm{\beta}}_{k-1}; \bm{I}_{\textbf{x},1:k-1}, z_{1:k-1}) + \nabla^2 g_k (\hat{\bm{\beta}}_k; \bm{I}_{\textbf{x},k}, z_k) \nonumber \\
H_0 (\hat{\bm{\beta}}_{0}) & = \varsigma I. \label{eqn:approx_ONM2}
\end{align}
The terms $G_k$ and $H_k$ represent approximate gradients and Hessians respectively for the cost function \eqref{eqn:cost_fn}.
The initialization $H_0 (\hat{\bm{\beta}}_{0}) = \varsigma I$  is a Levenberg-Marquardt type modification \cite{ChongZak} to ensure that the matrices $\{H_k\}$ are always non-singular.\footnote{In \cite{LeongZamaniShames} this is equivalently expressed as a full rank initialization  on $\left(H_0(\hat{\bm{\beta}}_{0})\right)^{-1}$.}

In the case where the field (and hence the parameters $\bm{\beta})$ is time-varying, the algorithm \eqref{eqn:approx_ONM1}-\eqref{eqn:approx_ONM2} may not be able to respond quickly to changes in $\bm{\beta}$, due to all past Hessians (including Hessians from old fields) being used in the computation of $H_k(\hat{\bm{\beta}}_{k}; \bm{I}_{\textbf{x},1:k}, z_{1:k}) $ in \eqref{eqn:approx_ONM2}. To overcome this problem, we  will introduce a \emph{forgetting factor} \cite{ManolakisIngleKogon} into the algorithm, 
where the forgetting factor $\delta$ satisfies $0 < \delta \leq 1$, and typically chosen to be close to one. The final estimation procedure is summarized as Algorithm~\ref{alg:DCT_optim_time_varying}. Compared to \eqref{eqn:approx_ONM2}, we note that the Levenberg-Marquardt modification in Algorithm \ref{alg:DCT_optim_time_varying} is done at every time step by adding $\varsigma I$ to $\tilde{H}_k$, as we found that only doing it once at the beginning can lead to algorithm instability due to exponential decay of initial conditions when using a forgetting factor. We also remark that Algorithm \ref{alg:DCT_optim_time_varying} reduces to \eqref{eqn:approx_ONM1}-\eqref{eqn:approx_ONM2} when the forgetting factor $\delta = 1$. 

\begin{algorithm}
\caption{Estimation of Fourier components using online optimization approach}
\label{alg:DCT_optim_time_varying}
\begin{algorithmic}[1]
\State \textbf{Algorithm Parameters}:  Logistic function parameter $\eta > 0$, Levenberg-Marquardt parameter $\varsigma > 0$, forgetting factor $\delta \in (0,1]$
\State \textbf{Inputs}: Initial position index $\bm{I}_{\textbf{x},0}$
\State \textbf{Outputs}: Parameter estimates $\{ \hat{\bm{\beta}}_k \}$
\State Initialize $\tilde{H}_0(\hat{\bm{\beta}}_{0}) = \mathbf{0}$
\For{$k=0,1,2,\dots,$}
	\State Update estimates 
 \begin{align*}
\hat{\bm{\beta}}_{k+1} &= \hat{\bm{\beta}}_k - \left( H_k (\hat{\bm{\beta}}_k; \bm{I}_{\textbf{x},1:k}, z_{1:k}) \right)^{-1} G_k (\hat{\bm{\beta}}_k; \bm{I}_{\textbf{x},k}, z_k) \nonumber \\
G_k (\hat{\bm{\beta}}_k; \bm{I}_{\textbf{x},k}, z_k) & = \nabla g_k (\hat{\bm{\beta}}_k; \bm{I}_{\textbf{x},k}, z_k) \nonumber \\
\tilde{H}_k(\hat{\bm{\beta}}_{k}; \bm{I}_{\textbf{x},1:k}, z_{1:k}) & =  \delta \tilde{H}_{k-1}(\hat{\bm{\beta}}_{k-1}; \bm{I}_{\textbf{x},1:k-1}, z_{1:k-1}) + \nabla^2 g_k (\hat{\bm{\beta}}_k; \bm{I}_{\textbf{x},k}, z_k) \nonumber \\
H_k(\hat{\bm{\beta}}_{k}; \bm{I}_{\textbf{x},1:k}, z_{1:k}) & = \tilde{H}_k(\hat{\bm{\beta}}_{k}; \bm{I}_{\textbf{x},1:k}, z_{1:k})  + \varsigma I, \label{eqn:approx_ONM_time_varying}
\end{align*}
\,\,\,\,\,\,\, where $\nabla g_k (\bm{\cdot}; \bm{\cdot}, \bm{\cdot})$ and $ \nabla^2 g_k(\bm{\cdot}; \bm{\cdot}, \bm{\cdot})$ are computed using \eqref{eqn:per_stage_gradient}-\eqref{eqn:per_stage_Hessian}	
	
 \State Determine  $\bm{I}_{\textbf{x},k+1} =\texttt{ActiveSensing}(\bm{I}_{\textbf{x},k}, \hat{\bm{\beta}}_{k+1})$  using Algorithm \ref{alg:active_sensing_online_opt}
\EndFor
\end{algorithmic}
\end{algorithm} 

\subsection{Measurement Location Selection Using Active Sensing}
For choosing the positions in which the unmanned vehicle should take measurements from, an ``active sensing'' approach \cite{Kreucher_active_sensing,LaShengChen,RisticSkvortsovGunatilaka} can be used, which aims to cleverly choose the next position given information collected so far, in order to more quickly obtain a good estimate of the field. 

In the case of binary measurements, a method for choosing the next measurement location is proposed in \cite{LeongZamaniShames}, that tries to maximize the minimum eigenvalue of an ``expected Hessian'' term $H^+(\bm{I}_{\textbf{x}'})$ over candidate future position indices $\bm{I}_{\textbf{x}'}$. Formally, the problem is:
$$\bm{I}_{\textbf{x},k+1} = \textnormal{arg} \max\limits_{\bm{I}_{\textbf{x}'} \in \mathcal{I}_{k+1}} \lambda_{\min} (H^+(\bm{I}_{\textbf{x}'})),$$
where $\lambda_{\min} (H^+(\bm{I}_{\textbf{x}'}))$ is the minimum eigenvalue of  $H^+(\bm{I}_{\textbf{x}'})$, $\mathcal{I}_{k+1}$ is the set of possible future position indices\footnote{The set $ \mathcal{I}_{k+1}$ may, e.g., capture the set of reachable positions from the current state of the mobile sensor platform.}
and
\begin{equation}
\label{eqn:expected_Hessian_binary}
\begin{split}
 H^+(\bm{I}_{\textbf{x}'}) & \triangleq H_k(\bm{\hat{\beta}}_{k}; \bm{I}_{\textbf{x},1:k}, z_{1:k})  +  \frac{ \eta^2 \exp(\eta (\bm{\hat{\beta}}_{k+1}^T \mathbf{K}(\bm{I}_{\textbf{x}'}) -  \tau)) }{\big(1+\exp(\eta (\bm{\hat{\beta}}_{k+1}^T \mathbf{K}(\bm{I}_{\textbf{x}'}) -  \tau)) \big)^2}  \mathbf{K}(\bm{I}_{\textbf{x}'}) \mathbf{K}(\bm{I}_{\textbf{x}'})^T \mathbb{P}(z' = 0) \\
 & \quad +  \frac{ \eta^2 \exp(\eta (\bm{\hat{\beta}}_{k+1}^T \mathbf{K}(\bm{I}_{\textbf{x}'}) -  \tau)) }{\big(1+\exp(\eta (\bm{\hat{\beta}}_{k+1}^T \mathbf{K}(\bm{I}_{\textbf{x}'}) -  \tau)) \big)^2}  \mathbf{K}(\bm{I}_{\textbf{x}'}) \mathbf{K}(\bm{I}_{\textbf{x}'})^T \mathbb{P}(z' = 1) \\
 & = H_k  (\bm{\hat{\beta}}_{k}; \bm{I}_{\textbf{x},1:k}, z_{1:k}) +  \frac{ \eta^2 \exp(\eta (\bm{\hat{\beta}}_{k+1}^T \mathbf{K}(\bm{I}_{\textbf{x}'}) -  \tau)) }{\big(1+\exp(\eta (\bm{\hat{\beta}}_{k+1}^T \mathbf{K}(\bm{I}_{\textbf{x}'}) -  \tau)) \big)^2}  \mathbf{K}(\bm{I}_{\textbf{x}'}) \mathbf{K}(\bm{I}_{\textbf{x}'})^T, 
\end{split}
\end{equation}
The last line of \eqref{eqn:expected_Hessian_binary} holds since $ \mathbb{P}(z' = 0)  +  \mathbb{P}(z' = 1)  = 1$, irrespective of the distribution of the noise $n(\bm{\cdot},\bm{\cdot})$. 

If we attempt to generalize \eqref{eqn:expected_Hessian_binary} to multi-level measurements, we find that there will be terms $ \mathbb{P}(z' = 0), \mathbb{P}(z' = 1), \dots,  \mathbb{P}(z' = L-1)  $ which cannot all be cancelled, and we will need to specify a noise distribution in order to compute these terms. Since exact knowledge of the noise distribution is usually unavailable in practice, we will instead consider a slightly different objective to optimize, namely a ``predicted Hessian''
\begin{equation}
\label{eqn:predicted_Hessian}
\begin{split}
 \hat{H}(\bm{I}_{\textbf{x}'}) & \triangleq H_k(\bm{\hat{\beta}}_{k}; \bm{I}_{\textbf{x},1:k}, z_{1:k})  +  \nabla^2 g_{k+1}(\bm{\hat{\beta}}_{k+1}; \bm{I}_{\textbf{x}'}, \hat{z}')
\end{split}
\end{equation}
where 
$\hat{z}' \triangleq q\big(\bm{\hat{\beta}}_{k+1}^T \mathbf{K}(\bm{I}_{\textbf{x}'})\big)$ is the predicted future measurement, with the quantizer $q(\bm{\cdot})$ given by \eqref{eqn:quantizer}. Note that \eqref{eqn:predicted_Hessian} reduces to \eqref{eqn:expected_Hessian_binary} in the case of binary measurements. We then maximize the minimum eigenvalue of the predicted Hessian to determine the next measurement location target:
\begin{equation}
\label{prob:maxmin_eig}
\bm{I}_{\textbf{x}}^{\textnormal{target}} = \textnormal{arg} \max\limits_{\bm{I}_{\textbf{x}'} \in \mathcal{I}_{k+1}} \lambda_{\min} (\hat{H}(\bm{I}_{\textbf{x}'})).
\end{equation}

For the set of candidate future position indices  $ \mathcal{I}_{k+1}$, one possible choice could be positions distributed uniformly on a grid within the search region $\mathcal{S}$. 
Once a new location target $\bm{I}_{\textbf{x}}^{\textnormal{target}} $ has been determined, the vehicle heads in that direction. The vehicle will collect measurements and update $\hat{\bm{\beta}}$ along the way, where we collect a new measurement after every $\rho_0$ in distance has been travelled until $ \bm{I}_{\textbf{x}}^{\textnormal{target}}$ is reached, at which time a new location target is determined. 
The procedure is summarized in Algorithm \ref{alg:active_sensing_online_opt},  where $\bm{I}_{\textnormal{closest}} (\mathbf{x})$ denotes the closest position index $(I_x, I_y)$ to $\mathbf{x} \in \mathcal{S}$.
The condition in line \ref{line:at_target} of Algorithm \ref{alg:active_sensing_online_opt} means the location target has been reached, so that a new location target is determined and a location index $\bm{I}_{\textbf{x},k+1}$ in the direction of the new target is returned. Some random exploration is also included in the algorithm, such that the new location target is random with probability $\varepsilon$, similar to $\varepsilon$-greedy algorithms used in reinforcement learning \cite{SuttonBarto}.  The condition in line \ref{line:within_range_target} means that the vehicle is within $\rho_0$ of the target, which will be reached at the next time step, while the condition in line \ref{line:outside_range_target} means  the vehicle will continue heading towards the target and collect measurements along the way. 

%An alternative method is used in \cite{LeongZamaniShames}, where a convex combination of the old and new directions is then determined, and a distance of $\rho_0$ is travelled when a new measurement is collected, $\hat{\bm{\beta}}$ is updated, and problem \eqref{prob:maxmin_eig} solved again. 

\begin{algorithm}[t]
\caption{Active sensing algorithm for online optimization approach: $\bm{I}_{\textbf{x},k+1} = \texttt{ActiveSensing}(\bm{I}_{\textbf{x},k}, \hat{\bm{\beta}}_{k+1})$}
\label{alg:active_sensing_online_opt}
\begin{algorithmic}[1]
\State \textbf{Algorithm Parameters}: Distance $\rho_0 \geq 0$, candidate position indices $\mathcal{I}_{k+1}$, search region $\mathcal{S}$, exploration probability $\varepsilon$
\State \textbf{Inputs}:    $\bm{I}_{\textbf{x},k}$, $\hat{\bm{\beta}}_{k+1}$
\State \textbf{Output}: Next position index $\bm{I}_{\textbf{x},k+1}$
\If{$k=0$}
    \State Initialize $\bm{I}_{\textbf{x}}^{\textnormal{target}}  = \bm{I}_{\textbf{x},0}$
\EndIf
\If{$\bm{I}_{\textbf{x},k} = \bm{I}_{\textbf{x}}^{\textnormal{target}}$} \label{line:at_target}
    \State With probability $\varepsilon$, set new $\bm{I}_{\textbf{x}}^{\textnormal{target}}$ to a random location index in $\{0, \dots, N_x - 1\} \times \{0, \dots, N_y-1\}$, otherwise compute new $\bm{I}_{\textbf{x}}^{\textnormal{target}} = \textnormal{arg} \max\limits_{\bm{I}_{\textbf{x}'} \in \mathcal{I}_{k+1}} \lambda_{\min} (\hat{H}(\bm{I}_{\textbf{x}'})),$ where $\hat{H}(\bm{I}_{\textbf{x}'})$ is given by \eqref{eqn:predicted_Hessian}  \label{line:random_exploration}
    \State Set $\mathbf{x}_{k+1} = \mathbf{x}_k + \rho_0 (\mathbf{x}^\textnormal{target} - \mathbf{x}_k)/||\mathbf{x}^{\textnormal{target}} - \mathbf{x}_k||$ and return $\bm{I}_{\textbf{x},k+1} = \bm{I}_{\textnormal{closest}}(\mathbf{x}_{k+1})$
\ElsIf{$||\mathbf{x}_k - \mathbf{x}^{\textnormal{target}}|| < \rho_0$} \label{line:within_range_target}
    \State Set $\mathbf{x}_{k+1} = \mathbf{x}^{\textnormal{target}}$ and return $\bm{I}_{\textbf{x},k+1} =\bm{I}_{\textbf{x}}^{\textnormal{target}}$
\Else \label{line:outside_range_target}
    \State Set $\mathbf{x}_{k+1} = \mathbf{x}_k + \rho_0 (\mathbf{x}^\textnormal{target} - \mathbf{x}_k)/||\mathbf{x}^{\textnormal{target}} - \mathbf{x}_k||$ and return $\bm{I}_{\textbf{x},k+1} = \bm{I}_{\textnormal{closest}}(\mathbf{x}_{k+1})$
\EndIf

\end{algorithmic}
\end{algorithm} 

\subsection{Refinement of Field Model}
\label{sec:field_refinement}
At the end of Section \ref{sec:DCT_RBF_comparison} we mentioned that we can refine our field model as we go along by including more higher order modes, while reusing previous estimates of the lower order modes. One reason for doing so could occur if we realize that the original set of modes chosen is not sufficient to provide an adequate estimate of the field, so that more modes need to be added. In this subsection we briefly describe how this refinement can be done. 

Suppose that originally, modes in $\tilde{\mathcal{U}}$ of cardinality $\tilde{\mathcal{N}}$ are being estimated, with an ordering on $\tilde{\mathcal{U}}$ indexed by $ j \in \{0, \dots, \tilde{\mathcal{N}} - 1 \}$. After $k$ iterations, suppose we wish to increase the set of estimated modes to the set $\tilde{\mathcal{U}}' \supseteq \tilde{\mathcal{U}}$, of cardinality $\tilde{\mathcal{N}}'$, and define an ordering on  $\tilde{\mathcal{U}}'$ indexed by $ j' \in \{0, \dots, \tilde{\mathcal{N}}' - 1 \}$. 

For the indices $j \in \{0, \dots, \tilde{\mathcal{N}} - 1 \}$, define a mapping 
$$m(.):\{0, \dots, \tilde{\mathcal{N}} - 1 \} \rightarrow \{0, \dots, \tilde{\mathcal{N}}' - 1 \}$$ such that $m(j) = j'$ gives the corresponding index $j' \in \{0, \dots, \tilde{\mathcal{N}}' - 1 \}$. Note that in general $m(.)$ may not be surjective. 

Let $\hat{\bm{\beta}}'$ denote the estimate  (of dimension $\tilde{N}'$), and $\tilde{H}'$ the matrix (of dimension $\tilde{N}' \times \tilde{N}'$) used in the computation of the approximate Hessians, for this new set of modes. 
Then we can re-initialize the estimates $\hat{\bm{\beta}}'$ by setting
$$ \hat{\bm{\beta}}'_{m(j),k+1} =  \hat{\bm{\beta}}_{j,k+1}, \quad \forall j \in \{0, \dots, \tilde{\mathcal{N}} - 1 \},
$$
which copies the estimates of the existing components $\hat{\bm{\beta}}$ (of dimension $\tilde{\mathcal{N}}$) across, with the other (new) components to be initialized appropriately. The term $\tilde{H}'_k$ should also have components corresponding to the existing set of modes copied across, by setting
$$ \tilde{H}'_{m(i) m(j),k} =  \tilde{H}_{i j,k}, \quad \forall i, j \in \{0, \dots, \tilde{\mathcal{N}} - 1 \},
$$
with other components of $\tilde{H}'_k$ set to zero. After this re-initialization, Algorithm~\ref{alg:DCT_optim_time_varying} then proceeds as before. 

\begin{remark}
Instead of adding more modes, the case where we refine the field model by deleting some modes can also be handled in a similar manner. 
\end{remark}

\section{Numerical Studies}
\label{sec:numerical}
For performance evaluation of the field estimation algorithms, we will consider two performance measures, the mean squared error (MSE) and structural similarity index (SSIM). These are defined similar to Section \ref{sec:DCT_RBF_comparison}, except that we replace the approximated field with the estimated field 
$$\hat{\phi}_d (I_x, I_y)   \triangleq \sum_{(u,v) \in \tilde{\mathcal{U}} } \alpha_x(u) \alpha_y(v) \hat{C}(u,v) \cos \left(\frac{(2I_x+1)\pi u}{2 N_x} \right) \cos \left(\frac{(2I_y+1)\pi v}{2 N_y} \right).$$

\subsection{Static Fields}
We consider estimation of the (true) field shown in Fig. \ref{fig:true_field_seed355}, with search region $\mathcal{S} = [0,100] \times [0,100]$. The field is discretized using $N_x = 100$ and $N_y = 100$. We use \eqref{eqn:U_tilde_largest} to select the largest modes that we wish to retain and estimate. 
\begin{figure}[t!]
\centering 
\includegraphics[scale=0.6]{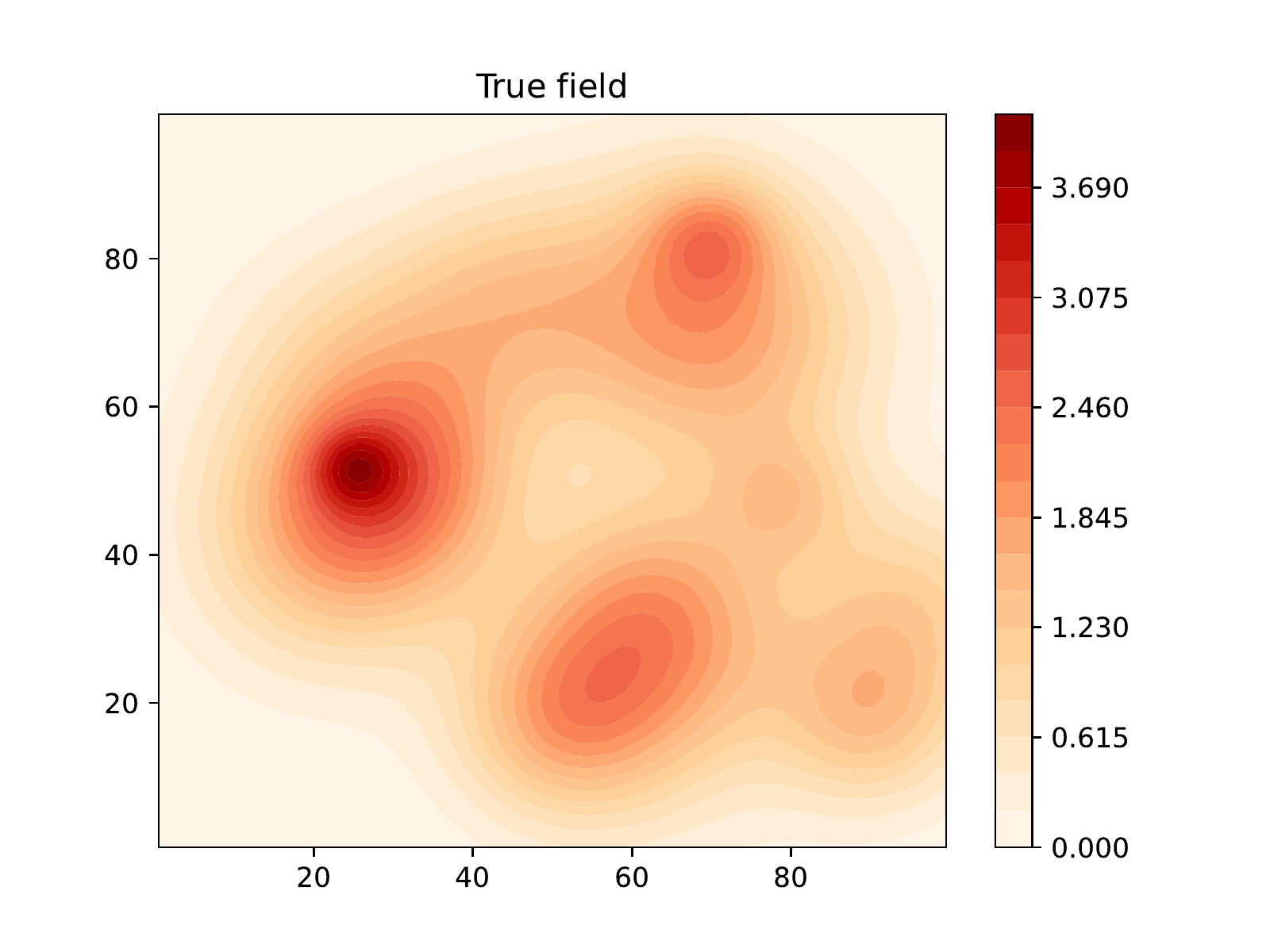} 
\caption{Static field}
\label{fig:true_field_seed355}
\end{figure} 

We use Algorithm \ref{alg:DCT_optim_time_varying} with $\eta=5$ and $\varsigma = 1/20000$. As the field is assumed static, the forgetting factor is set to $\delta = 1$. The initial position index is set to $\bm{I}_{\textbf{x},0} = (50,50)$, close to the center of the search region~$\mathcal{S}$. 
A four level quantizer is used with quantizer thresholds $\tau_0=1, \tau_1=2, \tau_2 = 3$. The measurement noise $n(\bm{\cdot},\bm{\cdot})$ is i.i.d. Gaussian with zero mean and variance equal to 0.1. For choosing the measurement locations, we use Algorithm \ref{alg:active_sensing_online_opt} with $\rho_0 = 10$. The candidate position indices $\mathcal{I}_{k+1}$ are chosen to correspond to 36 points placed uniformly on a grid within the search region $\mathcal{S}$. The exploration probability is chosen as $\varepsilon = 0.1$.

\begin{figure}[t!]
\centering 
\includegraphics[scale=0.6]{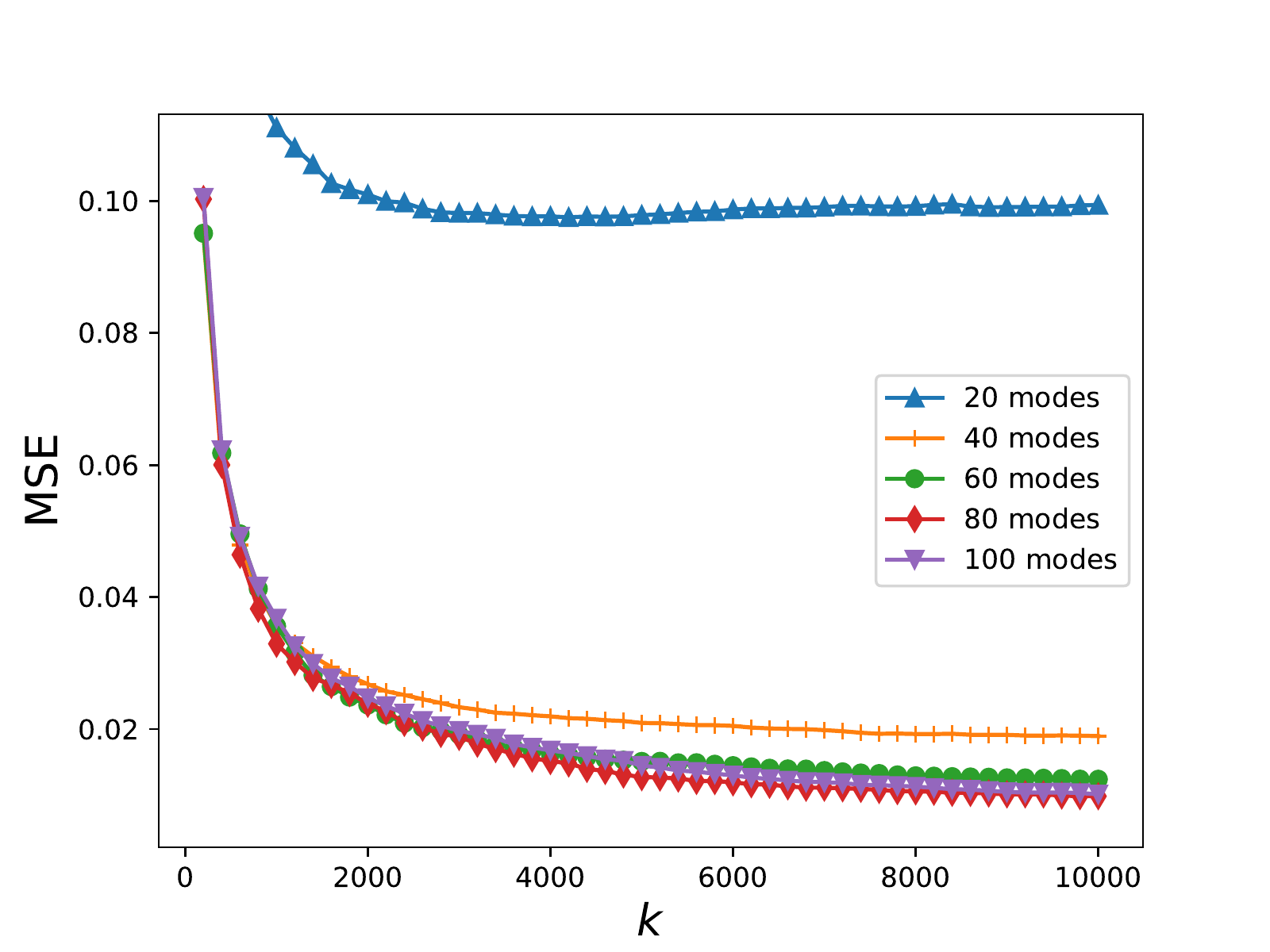} 
\caption{Static field: MSE vs. $k$}
\label{fig:MSE_time_stepped_seed355}
\end{figure} 

\begin{figure}[t!]
\centering 
\includegraphics[scale=0.6]{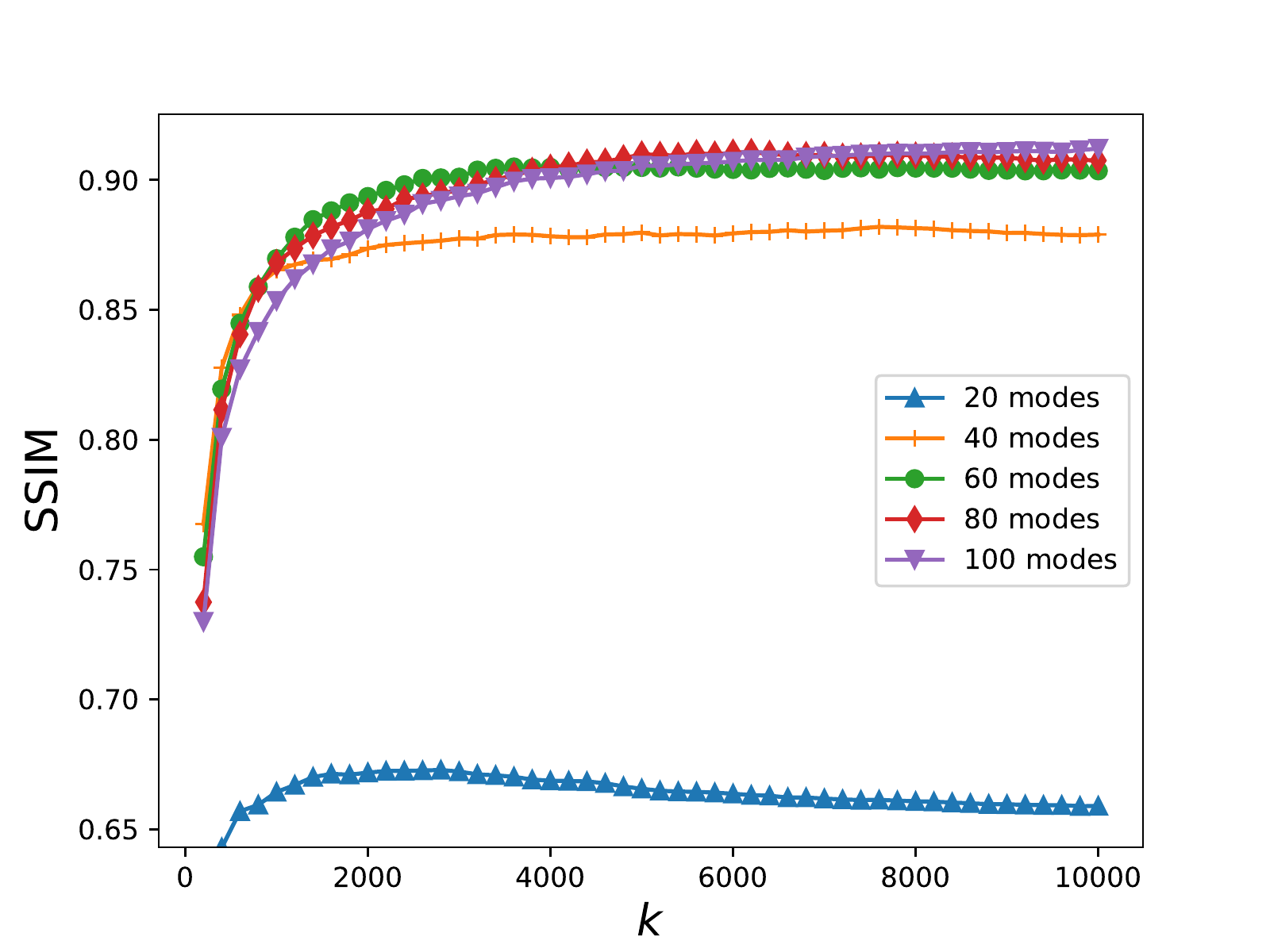} 
\caption{Static field: SSIM vs. $k$}
\label{fig:SSIM_time_stepped_seed355}
\end{figure} 

Fig. \ref{fig:MSE_time_stepped_seed355}  shows the MSE vs. $k$ (corresponding to the number of measurements collected), when various numbers of modes are estimated. Fig.  \ref{fig:SSIM_time_stepped_seed355} shows the SSIM vs. $k$. Each point in Figs. \ref{fig:MSE_time_stepped_seed355} and \ref{fig:SSIM_time_stepped_seed355} is obtained by averaging over 10 runs. We see from the figures that there is a trade-off between the estimation quality, number of modes/parameters that need to be estimated, and number of measurements collected. If a lot of measurements can be collected, then estimating more modes will allow for a better estimate of the field.\footnote{For example, if multiple vehicles can be utilized \cite{LeongZamani_SP} or one has a sensor network, then more measurements can be collected in a limited amount of time.} On the other hand, if fewer measurements are available, estimating fewer modes more accurately may give a better field estimate than estimating lots of modes inaccurately.  
In Fig. \ref{fig:estimated_field_seed355} we show a sample plot of the estimated field when 60 modes are estimated, after 2000 measurements have been collected. 
\begin{figure}[t!]
\centering 
\includegraphics[scale=0.6]{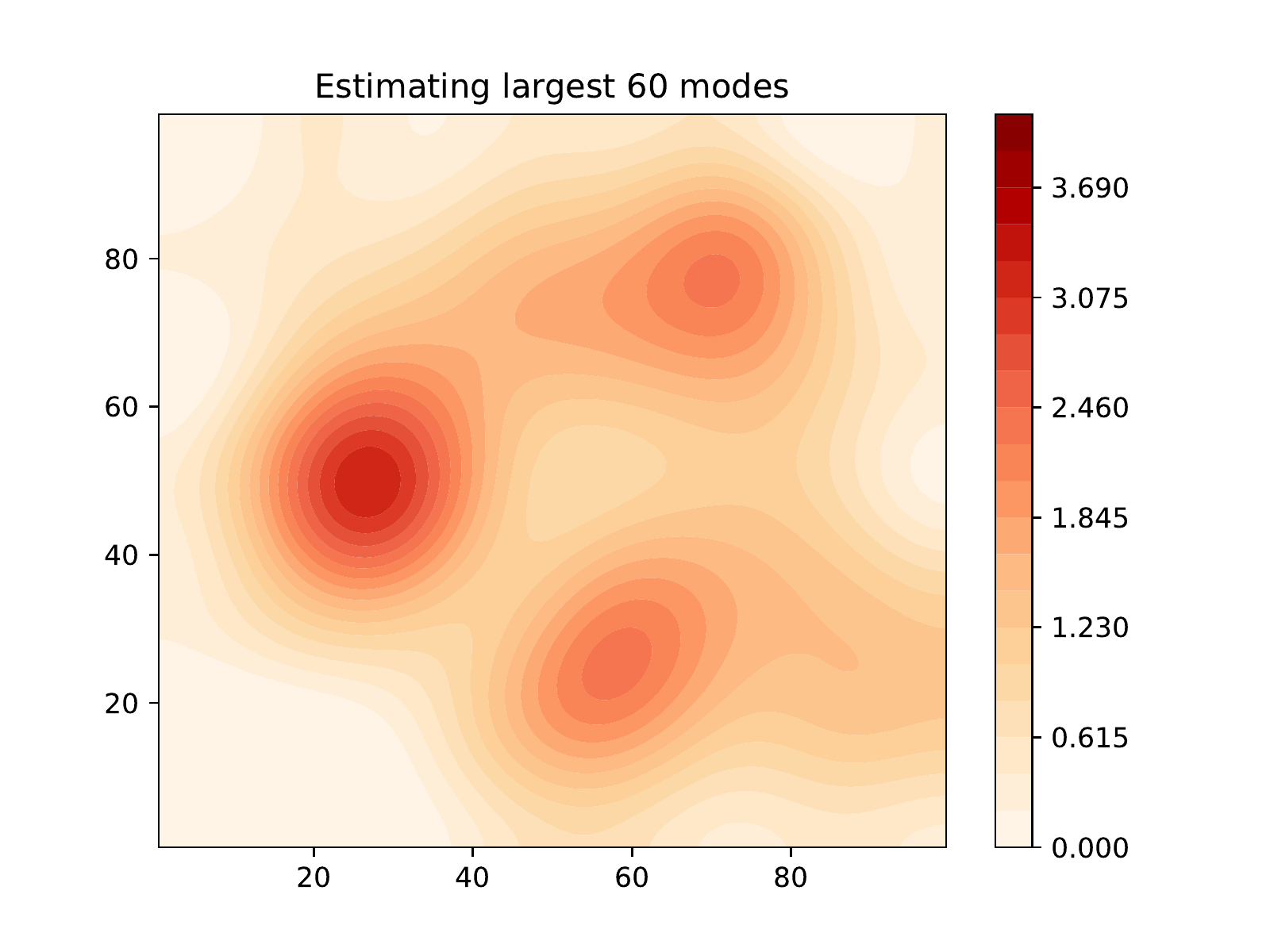} 
\caption{Static field: Estimated field using 2000 measurements}
\label{fig:estimated_field_seed355}
\end{figure} 

\subsection{Time-varying Fields}
We now consider an example with time-varying fields. Suppose the true field is the same of that of Fig. \ref{fig:field_seed341_DCT} for the first 1000 iterations, but then switches to the true field in Fig. \ref{fig:field_seed343_DCT} for the next 1000 iterations. We will use Algorithms \ref{alg:DCT_optim_time_varying} and \ref{alg:active_sensing_online_opt}  with forgetting factor $\delta = 0.995$, with other parameters the same as in the previous example. 

Figs. \ref{fig:MSE_time_varying_seed341_343} and \ref{fig:SSIM_time_varying_seed341_343} show respectively the MSE and SSIM vs. $k$, when the 60 largest modes are estimated. 
We see that after the field changes at $k=1000$ the accuracy of the field estimate drops, but Algorithm~\ref{alg:DCT_optim_time_varying} is able to recover  and estimate the new field as more measurements are collected. 

For comparison, the MSE and SSIM obtained using forgetting factor $\delta = 1$ are also shown. In this case, as there is no forgetting of old information, the field estimates will take much longer to adjust to the new field. 

\begin{figure}[t!]
\centering 
\includegraphics[scale=0.6]{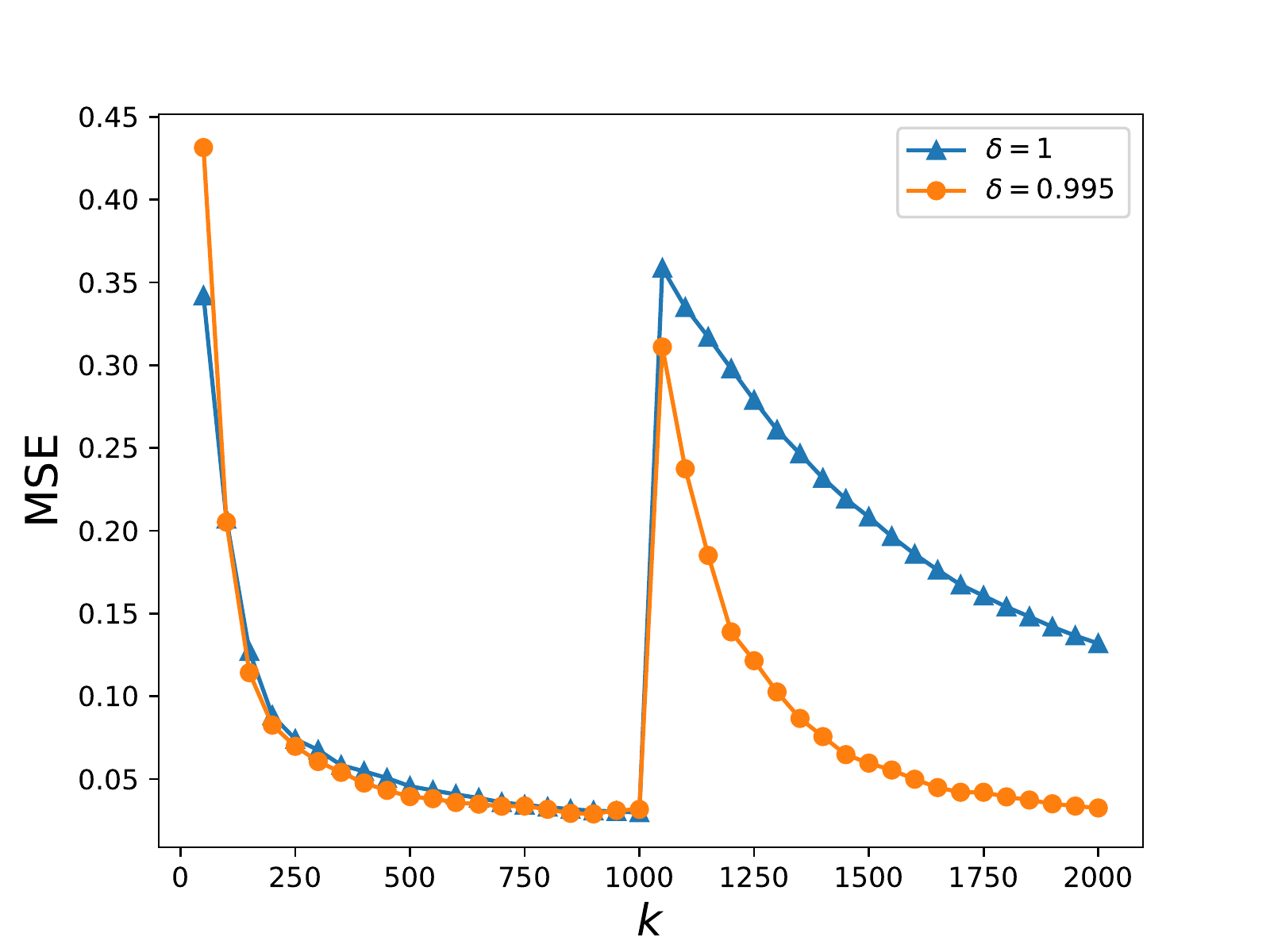} 
\caption{Time varying field: MSE vs. $k$}
\label{fig:MSE_time_varying_seed341_343}
\end{figure} 

\begin{figure}[t!]
\centering 
\includegraphics[scale=0.6]{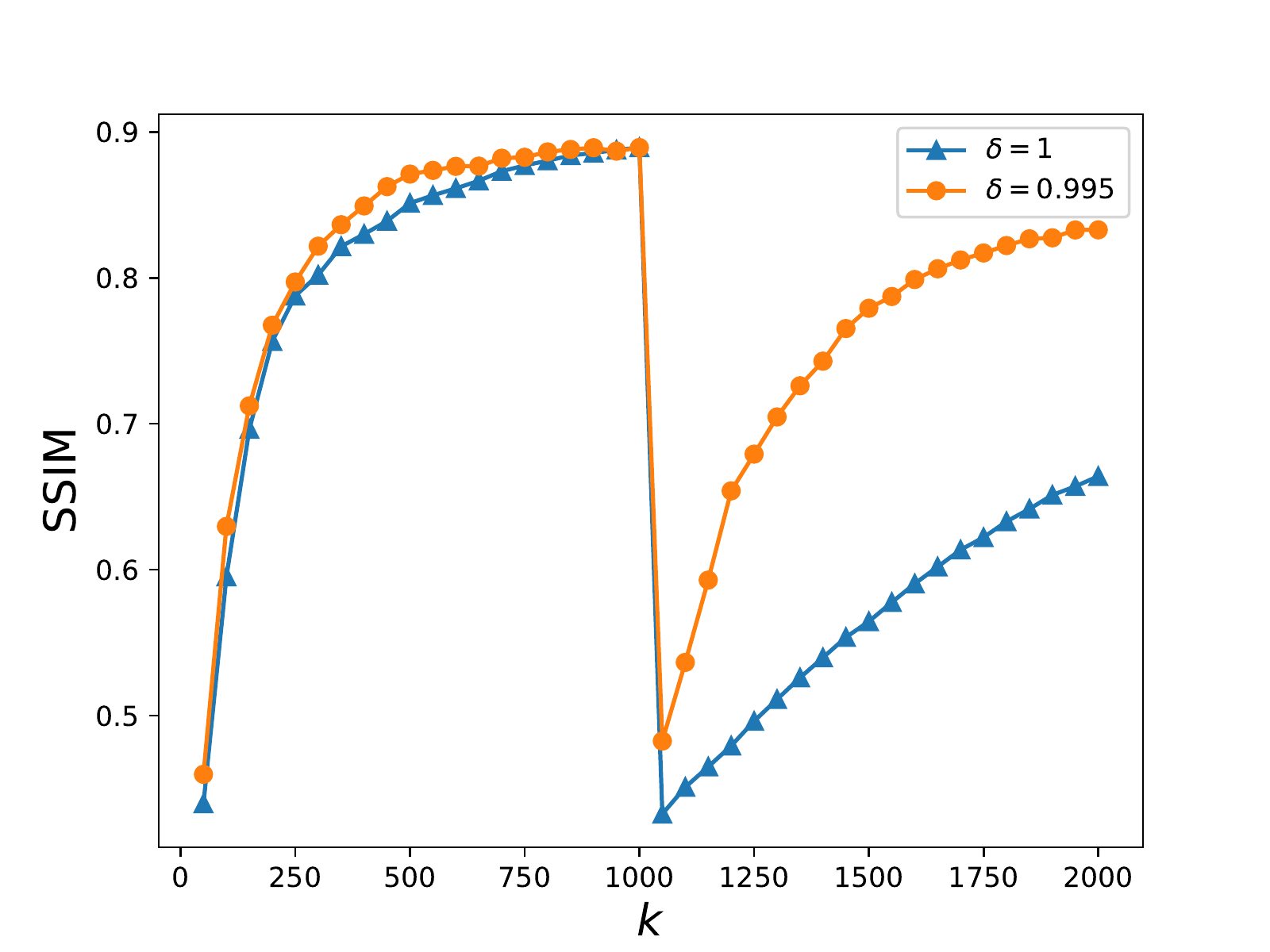} 
\caption{Time varying field: SSIM vs. $k$}
\label{fig:SSIM_time_varying_seed341_343}
\end{figure} 

\subsection{Refinement of Field Model}
Here we consider the use of a refined field model as in Section \ref{sec:field_refinement}, for estimation of the static field shown in Fig. \ref{fig:true_field_seed355}. For the first 1000 iterations the 40 largest modes are estimated, while the 80 largest modes are estimated for the next 1000 iterations, with the estimates of the 40 largest modes reused as described in Section \ref{sec:field_refinement}. Figs. \ref{fig:MSE_refinement_seed355} and \ref{fig:SSIM_refinement_seed355} show respectively the MSE and SSIM vs. $k$. We see that just after $k=1000$ the estimate quality using 80 modes decreases slightly, due to the new modes not being accurately estimated, but the performance quickly improves and eventually outperforms the use of 40 modes. 

\begin{figure}[t!]
\centering 
\includegraphics[scale=0.6]{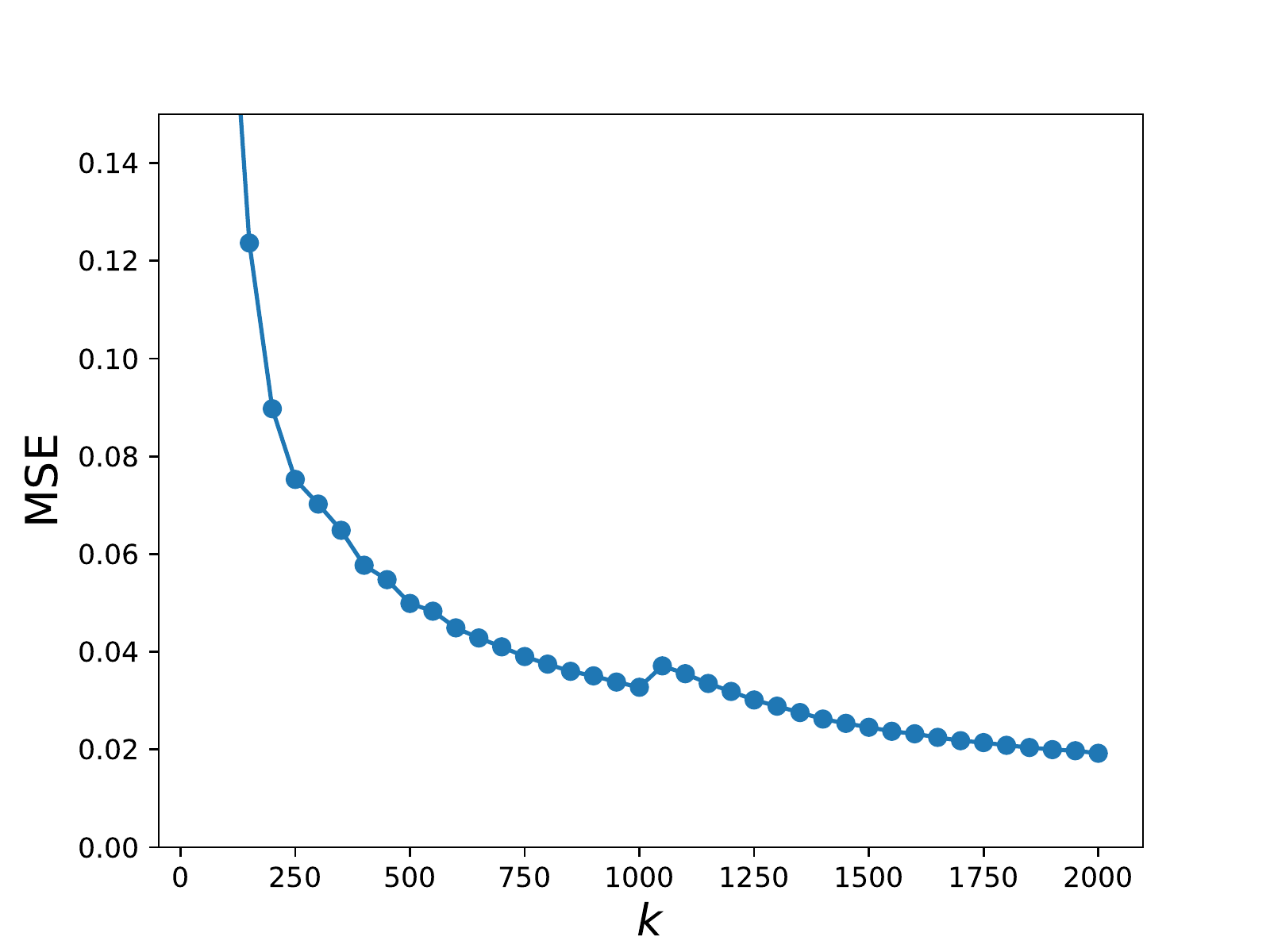} 
\caption{Estimating 40 modes for first 1000 iterations, then 80 modes for next 1000 iterations: MSE vs. $k$}
\label{fig:MSE_refinement_seed355}
\end{figure} 

\begin{figure}[t!]
\centering 
\includegraphics[scale=0.6]{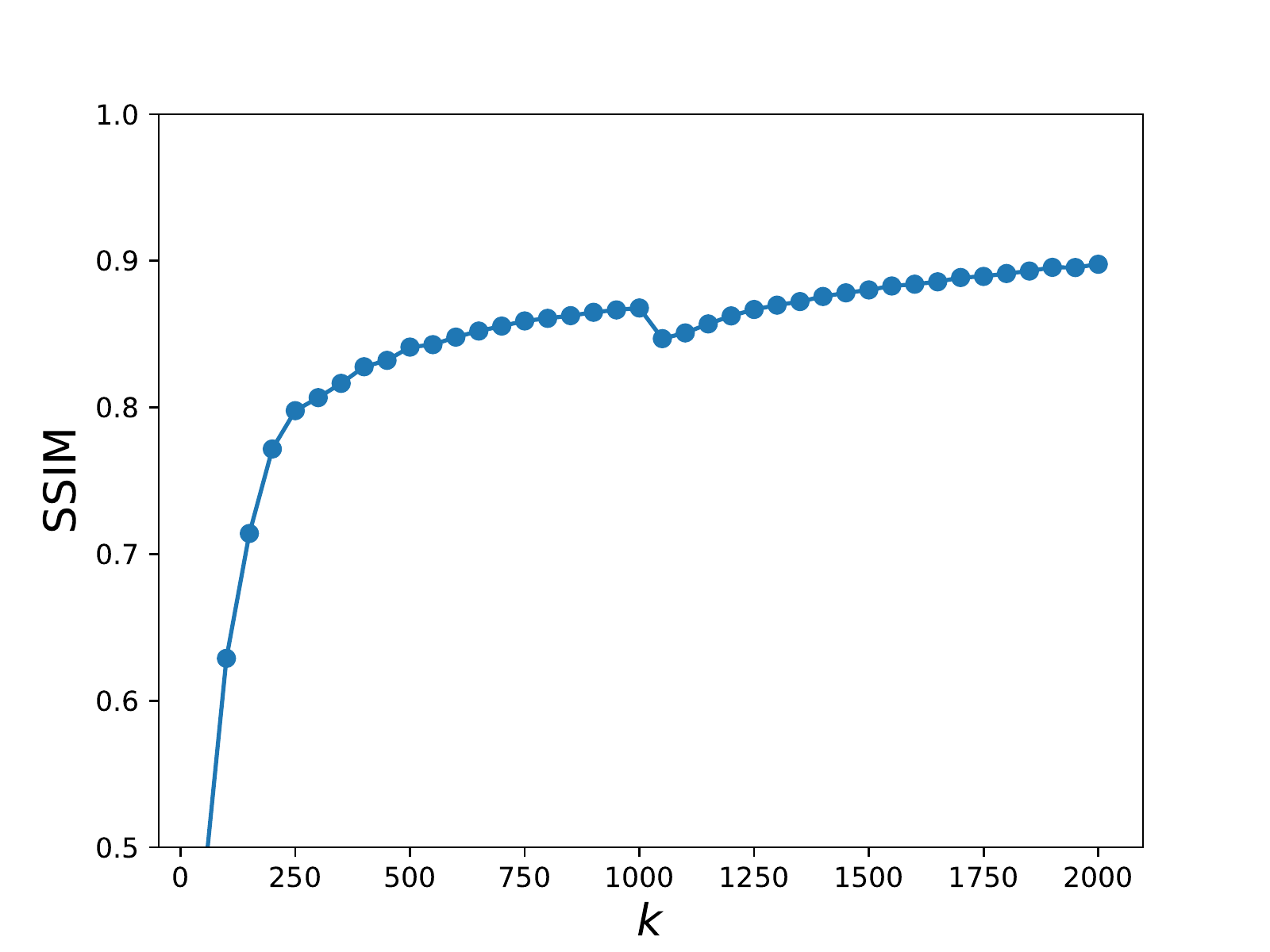} 
\caption{Estimating 40 modes for first 1000 iterations, then 80 modes for next 1000 iterations: SSIM vs. $k$}
\label{fig:SSIM_refinement_seed355}
\end{figure} 

\section{Conclusion}
This paper has studied the estimation of scalar fields, where the field is viewed in the Fourier domain. An algorithm has been presented for estimating the lower order modes of the field, under the assumption of noisy quantized measurements collected from the environment. Our approach assumed an unmanned autonomous vehicle travelling around a region in order to collect these measurements. The setup can also be extended to multiple vehicles, with the vehicles sharing measurements with each other similar to \cite{LeongZamani_SP}. Future work will consider the use of a sensor network for field estimation, with algorithms constrained by local communication and distributed computation.

\section*{Acknowledgment}
The authors thank Mr. Shintaro Umeki for suggesting the use of structural similarity as a performance measure while working at Defence Science and Technology Group as a summer vacation student. 

\begin{appendix}

\section{Proof of Lemma \ref{lemma:optimal_C_DCT}}
\label{appendix:optimal_C_DCT_proof}
By definition, 
%\begin{equation}
%\label{eqn:MSE_derivation_DCT}
%\begin{split}
\begin{align}
\textnormal{MSE} &= \frac{1}{N_x N_y} \sum_{I_x=0}^{N_x-1} \sum_{I_y=0}^{N_y-1} \Big( \phi_d (I_x, I_y) - \tilde{\phi}_d (I_x, I_y)  \Big)^2  \nonumber \\
& = \frac{1}{N_x N_y} \sum_{I_x=0}^{N_x-1} \sum_{I_y=0}^{N_y-1}  \Bigg( \sum_{(u,v) \in \mathcal{U}} \alpha_x(u) \alpha_y(v) C(u,v) \cos \left(\frac{(2I_x+1)\pi u}{2 N_x} \right) \cos \left(\frac{(2I_y+1)\pi v}{2 N_y} \right) \nonumber \\
& \quad - \sum_{(u,v) \in \tilde{\mathcal{U}} } \alpha_x(u) \alpha_y(v) \tilde{C}(u,v) \cos \left(\frac{(2I_x+1)\pi u}{2 N_x} \right) \cos \left(\frac{(2I_y+1)\pi v}{2 N_y} \right) \Bigg) ^2 \nonumber \\
&= \frac{1}{N_x N_y} \sum_{I_x=0}^{N_x-1} \sum_{I_y=0}^{N_y-1}  \Bigg( \sum_{(u,v) \in \tilde{\mathcal{U}}} \alpha_x(u) \alpha_y(v) \big(C(u,v) - \tilde{C}(u,v) \big) \nonumber  \\ 
& \quad \quad \times \cos \left(\frac{(2I_x+1)\pi u}{2 N_x} \right) \cos \left(\frac{(2I_y+1)\pi v}{2 N_y} \right) \nonumber \\
& \quad + \sum_{(u,v) \in \mathcal{U} \setminus \tilde{\mathcal{U}} } \alpha_x(u) \alpha_y(v) C(u,v) \cos \left(\frac{(2I_x+1)\pi u}{2 N_x} \right) \cos \left(\frac{(2I_y+1)\pi v}{2 N_y} \right) \Bigg) ^2 \nonumber \\
& = \frac{1}{N_x N_y} \sum_{I_x=0}^{N_x-1} \sum_{I_y=0}^{N_y-1}  \Bigg[ \Bigg( \sum_{(u,v) \in \tilde{\mathcal{U}}} \alpha_x(u) \alpha_y(v) \big(C(u,v) - \tilde{C}(u,v) \big) \nonumber \\ 
& \quad \quad \times \cos \left(\frac{(2I_x+1)\pi u}{2 N_x} \right) \cos \left(\frac{(2I_y+1)\pi v}{2 N_y} \right) \Bigg)^2 \nonumber \\
& \quad + 2 \Bigg( \sum_{(u,v) \in \tilde{\mathcal{U}}} \alpha_x(u) \alpha_y(v) \big(C(u,v) - \tilde{C}(u,v) \big) \cos \left(\frac{(2I_x+1)\pi u}{2 N_x} \right) \cos \left(\frac{(2I_y+1)\pi v}{2 N_y} \right) \Bigg) \nonumber \\
& \quad \quad \times \Bigg( \sum_{(u,v) \in \mathcal{U} \setminus \tilde{\mathcal{U}} } \alpha_x(u) \alpha_y(v) C(u,v) \cos \left(\frac{(2I_x+1)\pi u}{2 N_x} \right) \cos \left(\frac{(2I_y+1)\pi v}{2 N_y} \right) \Bigg) \nonumber \\
& \quad + \Bigg(\sum_{(u,v) \in \mathcal{U} \setminus \tilde{\mathcal{U}} } \alpha_x(u) \alpha_y(v) C(u,v) \cos \left(\frac{(2I_x+1)\pi u}{2 N_x} \right) \cos \left(\frac{(2I_y+1)\pi v}{2 N_y} \right) \Bigg)^2 \Bigg] \nonumber \\
& = \frac{1}{N_x N_y} \sum_{I_x=0}^{N_x-1} \sum_{I_y=0}^{N_y-1}  \Bigg[ \Bigg( \sum_{(u,v) \in \tilde{\mathcal{U}}} \alpha_x(u) \alpha_y(v) \big(C(u,v) - \tilde{C}(u,v) \big) \nonumber \\ 
& \quad \quad \times \cos \left(\frac{(2I_x+1)\pi u}{2 N_x} \right) \cos \left(\frac{(2I_y+1)\pi v}{2 N_y} \right) \Bigg)^2 \nonumber \\
& \quad + \Bigg(\sum_{(u,v) \in \mathcal{U} \setminus \tilde{\mathcal{U}} } \alpha_x(u) \alpha_y(v) C(u,v) \cos \left(\frac{(2I_x+1)\pi u}{2 N_x} \right) \cos \left(\frac{(2I_y+1)\pi v}{2 N_y} \right) \Bigg)^2 \Bigg]  \label{eqn:MSE_derivation_DCT}
\end{align}
%\end{split}
%\end{equation}
The last equality follows since 
\begin{align*}
\sum_{I_x=0}^{N_x-1} \sum_{I_y=0}^{N_y-1} & \alpha_x(u) \alpha_y(v) \big(C(u,v) - \tilde{C}(u,v) \big) \cos \left(\frac{(2I_x+1)\pi u}{2 N_x} \right) \cos \left(\frac{(2I_y+1)\pi v}{2 N_y} \right) \\
& \quad \times \alpha_x(u') \alpha_y(v') C(u',v') \cos \left(\frac{(2I_x+1)\pi u'}{2 N_x} \right) \cos \left(\frac{(2I_y+1)\pi v'}{2 N_y} \right)  
\end{align*}
is equal to zero for all $(u,v) \in \mathcal{U}$ and $ (u',v') \in \mathcal{U} \setminus \tilde{\mathcal{U}}$, by orthogonality of the DCT basis vectors \cite{AhmedNatarajanRao,Strang_DCT}. To conclude the proof, we note that the expression for the MSE given in the last equality of \eqref{eqn:MSE_derivation_DCT} is  clearly minimized when $ \tilde{C}(u,v) = C(u,v), \, \forall (u,v) \in \tilde{U}.$

\end{appendix}

\bibliography{IEEEabrv,source_localization}

% Generated by IEEEtran.bst, version: 1.14 (2015/08/26)
\begin{thebibliography}{10}
\providecommand{\url}[1]{#1}
\csname url@samestyle\endcsname
\providecommand{\newblock}{\relax}
\providecommand{\bibinfo}[2]{#2}
\providecommand{\BIBentrySTDinterwordspacing}{\spaceskip=0pt\relax}
\providecommand{\BIBentryALTinterwordstretchfactor}{4}
\providecommand{\BIBentryALTinterwordspacing}{\spaceskip=\fontdimen2\font plus
\BIBentryALTinterwordstretchfactor\fontdimen3\font minus
  \fontdimen4\font\relax}
\providecommand{\BIBforeignlanguage}[2]{{%
\expandafter\ifx\csname l@#1\endcsname\relax
\typeout{** WARNING: IEEEtran.bst: No hyphenation pattern has been}%
\typeout{** loaded for the language `#1'. Using the pattern for}%
\typeout{** the default language instead.}%
\else
\language=\csname l@#1\endcsname
\fi
#2}}
\providecommand{\BIBdecl}{\relax}
\BIBdecl

\bibitem{HutchinsonOh}
M.~Hutchinson, H.~Oh, and W.-H. Chen, ``A review of source term estimation
  methods for atmospheric dispersion events using static or mobile sensors,''
  \emph{Inf. Fusion}, vol.~36, pp. 130--148, 2017.

\bibitem{RisticMorelandeGunatilaka}
B.~Ristic, M.~Morelande, and A.~Gunatilaka, ``Information driven search for
  point sources of gamma radiation,'' \emph{Signal Process.}, vol.~90, pp.
  1225--1239, 2010.

\bibitem{Yardibi}
T.~Yardibi, J.~Li, P.~Stoica, M.~Xue, and A.~B. Baggeroer, ``Source
  localization and sensing: A nonparametric iterative adaptive approach based
  on weighted least squares,'' \emph{{IEEE} Trans. Aerosp. Electron. Syst.},
  vol.~46, no.~1, pp. 425--443, Jan. 2010.

\bibitem{AnnunzioYoungHaupt}
A.~J. Annunzio, G.~S. Young, and S.~E. Haupt, ``Utilizing state estimation to
  determine the source location for a contaminant,'' \emph{Atmos. Environ.},
  vol.~46, pp. 580--589, 2012.

\bibitem{NeumannBennetts_advanced_robotics}
P.~P. Neumann, V.~{Hernandez Bennetts}, A.~J. Lilienthal, M.~Bartholmai, and
  J.~H. Schiller, ``Gas source localization with a micro-drone using
  bio-inspired and particle filter-based algorithms,'' \emph{Advanced
  Robotics}, vol.~27, no.~9, pp. 725--738, 2013.

\bibitem{WadeSenocak}
D.~Wade and I.~Senocak, ``Stochastic reconstruction of multiple source
  atmospheric contaminant dispersion events,'' \emph{Atmos. Environ.}, vol.~74,
  pp. 45--51, 2013.

\bibitem{NewazJeong}
A.~A.~R. Newaz, S.~Jeong, H.~Lee, H.~Ryu, and N.~Y. Chong, ``{UAV}-based
  multiple source localization and contour mapping of radiation fields,''
  \emph{Robotics and Autonomous Systems}, vol.~85, pp. 12--25, 2016.

\bibitem{RisticGunatilakaGailis}
B.~Ristic, A.~Gunatilaka, and R.~Gailis, ``Localisation of a source of
  hazardous substance dispersion using binary measurements,'' \emph{Atmos.
  Environ.}, vol. 142, pp. 114--119, 2016.

\bibitem{Selvaratnam_CDC}
D.~D. Selvaratnam, I.~Shames, D.~V. Dimarogonas, J.~H. Manton, and B.~Ristic,
  ``Co-operative estimation for source localisation using binary sensors,'' in
  \emph{Proc. {IEEE} Conf. Decision and Control}, Melbourne, Australia, Dec.
  2017, pp. 1572--1577.

\bibitem{HutchinsonLiu}
M.~Hutchinson, C.~Liu, and W.-H. Chen, ``Source term estimation of a hazardous
  airborne release using an unmanned aerial vehicle,'' \emph{J. Field
  Robotics}, vol.~36, pp. 797--817, 2019.

\bibitem{EslingerMendez}
P.~W. Eslinger, J.~M. Mendez, and B.~T. Schrom, ``Source term estimation in the
  presence of nuisance signals,'' \emph{J. Environ. Radioact.}, vol. 203, pp.
  220--225, 2019.

\bibitem{LiChen}
D.~Li, F.~Chen, Y.~Wang, and X.~Wang, ``Implementation of a
  {UAV}-sensory-system-based hazard source estimation in a chemical plant
  cluster,'' in \emph{IOP Conf. Series}, 2019, p. 012043.

\bibitem{ParkAnSeoOh}
M.~Park, S.~An, J.~Seo, and H.~Oh, ``Autonomous source search for {UAVs} using
  {Gaussian} mixture model-based infotaxis: Algorithm and flight experiments,''
  \emph{{IEEE} Trans. Aerosp. Electron. Syst.}, vol.~57, no.~6, pp. 4238--4254,
  Dec. 2021.

\bibitem{WeidmannHirst}
D.~Weidmann, B.~Hirst, M.~Jones, R.~Ijzermans, D.~Randell, N.~Macleod,
  A.~Kannath, J.~Chu, and M.~Dean, ``Locating and quantifying methane emissions
  by inverse analysis of path-integrated concentration data using a
  {Markov-Chain Monte Carlo} approach,'' \emph{ACS Earth Space Chem.}, vol.~6,
  pp. 2190--2198, Jun. 2022.

\bibitem{MartinPayton}
P.~Martin, O.~Payton, J.~Fardoulis, D.~Richards, Y.~Yamashiki, and T.~Scott,
  ``Low altitude unmanned aerial vehicle for characterising remediation
  effectiveness following the {FDNPP} accident,'' \emph{J. Environ. Radioact.},
  vol. 151, pp. 58--63, Jun. 2016.

\bibitem{WangYangWu}
Z.~Wang, J.~Yang, and J.~Wu, ``Level set estimation of spatial-temporally
  correlated random fields with active sparse sensing,'' \emph{{IEEE} Trans.
  Aerosp. Electron. Syst.}, vol.~53, no.~2, pp. 862--876, Apr. 2017.

\bibitem{MorelandeSkvortsov}
M.~R. Morelande and A.~Skvortsov, ``Radiation field estimation using a
  {Gaussian} mixture,'' in \emph{Proc. Intl. Conf. Inf. Fusion}, Seattle, USA,
  Jul. 2009, pp. 2247--2254.

\bibitem{LaSheng}
H.~M. La and W.~Sheng, ``Distributed sensor fusion for scalar field mapping
  using mobile sensor networks,'' \emph{{IEEE} Trans. Cybern.}, vol.~43, no.~2,
  pp. 766--778, Apr. 2013.

\bibitem{LaShengChen}
H.~M. La, W.~Sheng, and J.~Chen, ``Cooperative and active sensing in mobile
  sensor networks for scalar field mapping,'' \emph{{IEEE} Trans. Syst., Man,
  Cybern., Syst.}, vol.~45, no.~1, pp. 1--12, Jan. 2015.

\bibitem{RazakSukumarChung_journal}
R.~A. Razak, S.~Sukumar, and H.~Chung, ``Scalar field estimation with mobile
  sensor networks,'' \emph{Int. J. Robust Nonlinear Control}, vol.~31, pp.
  4287--4305, 2021.

\bibitem{LeongZamani_SP}
A.~S. Leong and M.~Zamani, ``Field estimation using binary measurements,''
  \emph{Signal Process.}, vol. 194, no. 108430, 2022.

\bibitem{LeongZamaniShames}
A.~S. Leong, M.~Zamani, and I.~Shames, ``A logistic regression approach to
  field estimation using binary measurements,'' \emph{{IEEE} Signal Process.
  Lett.}, vol.~29, pp. 1848--1852, 2022.

\bibitem{TranGarratt}
V.~P. Tran, M.~A. Garratt, K.~Kasmarik, S.~G. Anavatti, A.~S. Leong, and
  M.~Zamani, ``Multi-gas source localization and mapping by flocking robots,''
  \emph{Inf. Fusion}, vol.~91, pp. 665--680, 2023.

\bibitem{BritanikYipRao}
V.~Britanik, P.~Yip, and K.~R. Rao, \emph{Discrete Cosine and Sine
  Transforms}.\hskip 1em plus 0.5em minus 0.4em\relax Academic Press, 2007.

\bibitem{Strang_DCT}
G.~Strang, ``The discrete cosine transform,'' \emph{{SIAM} Review}, vol.~41,
  no.~1, pp. 135--147, 1999.

\bibitem{YamataniSaito}
K.~Yamatani and N.~Saito, ``Improvement of {DCT}-based compression algorithms
  using {Poisson}'s equation,'' \emph{{IEEE} Trans. Image Process.}, vol.~15,
  no.~12, pp. 3672--3689, 2006.

\bibitem{RobinsRapleyThomas}
P.~Robins, V.~Rapley, and P.~Thomas, ``A probabilistic chemical sensor model
  for data fusion,'' in \emph{Proc. Int. Conf. Inf. Fusion}, Philadelphia, USA,
  Jul. 2005, pp. 1116--1122.

\bibitem{ChengKondaSinghScott}
Y.~Cheng, U.~Konda, T.~Singh, and P.~Scott, ``Bayesian estimation for {CBRN}
  sensors with non-{Gaussian} likelihood,'' \emph{{IEEE} Trans. Aerosp.
  Electron. Syst.}, vol.~47, no.~1, pp. 684--701, Jan. 2011.

\bibitem{CalafioreElGhaoui}
G.~Calafiore and L.~{El Ghaoui}, \emph{Optimization Models}.\hskip 1em plus
  0.5em minus 0.4em\relax Cambridge University Press, 2014.

\bibitem{Murphy_book1}
K.~P. Murphy, \emph{Probabilistic Machine Learning: An Introduction}.\hskip 1em
  plus 0.5em minus 0.4em\relax The {MIT} Press, 2022.

\bibitem{WangBovikSheikhSimoncelli}
Z.~Wang, A.~C. Bovik, H.~R. Sheikh, and E.~P. Simoncelli, ``Image quality
  assessment: From error visibility to structural similarity,'' \emph{{IEEE}
  Trans. Image Process.}, vol.~13, no.~4, pp. 600--612, Apr. 2004.

\bibitem{WangBovik_MSE}
Z.~Wang and A.~C. Bovik, ``Mean squared error: Love it or leave it?''
  \emph{{IEEE} Signal Process. Mag.}, vol.~26, no.~1, pp. 98--117, Jan. 2009.

\bibitem{LesageLandryTaylorShames}
A.~Lesage-Landry, J.~A. Taylor, and I.~Shames, ``Second-order online nonconvex
  optimization,'' \emph{{IEEE} Trans. Autom. Control}, vol.~66, no.~10, pp.
  4866--4872, Oct. 2021.

\bibitem{ChongZak}
E.~K.~P. Chong and S.~H. \.{Z}ak, \emph{An Introduction to Optimization},
  4th~ed.\hskip 1em plus 0.5em minus 0.4em\relax John Wiley \& Sons, 2013.

\bibitem{ManolakisIngleKogon}
D.~G. Manolakis, V.~K. Ingle, and S.~M. Kogon, \emph{Statistical and Adaptive
  Signal Processing}.\hskip 1em plus 0.5em minus 0.4em\relax Artech House,
  2005.

\bibitem{Kreucher_active_sensing}
C.~M. Kreucher, K.~D. Kastella, and A.~O. Hero, ``Sensor management using an
  active sensing approach,'' \emph{Signal Process.}, vol.~85, pp. 607--624,
  2005.

\bibitem{RisticSkvortsovGunatilaka}
B.~Ristic, A.~Skvortsov, and A.~Gunatilaka, ``A study of cognitive strategies
  for an autonomous search,'' \emph{Inf. Fusion}, vol.~28, pp. 1--9, 2016.

\bibitem{SuttonBarto}
R.~S. Sutton and A.~G. Barto, \emph{Reinforcement Learning}, 2nd~ed.\hskip 1em
  plus 0.5em minus 0.4em\relax The {MIT Press}, 2018.

\bibitem{AhmedNatarajanRao}
N.~Ahmed, T.~Natarajan, and K.~R. Rao, ``Discrete cosine transform,''
  \emph{{IEEE} Trans. Comput.}, vol. C-23, no.~1, pp. 90--93, Jan. 1974.

\end{thebibliography}
\bibliographystyle{IEEEtran}

\end{document}